\pdfoutput=1
\RequirePackage{boolean}
\RequirePackage{reversion}

\IfFileExists
    {\jobname.cfg}
    {\input{\jobname.cfg}}
    {\input{article.cfg}}

\LLNCS{}{
  \RequirePackage[backref=page,bookmarksnumbered]{hyperref}
}

\LLNCS{
  \documentclass[oribibl]{llncs}
}{
  \documentclass[acmlarge\Anonymous{,anonymous}{}\Draft{,review}{}]{acmart}
  \Draft{\settopmatter{printfolios=true}}{}
  \ACM{}{\settopmatter{printacmref=false}}

  \providecolor{BLUE}{named}{ACMBlue}
  \providecolor{RED}{named}{ACMRed}
}

\LLNCS{
  \usepackage[backref=page,bookmarksnumbered]{hyperref}
}{}

\usepackage[nohyperref,notitlecaps,\LLNCS{nosigplan,llncs}{acmart}]{mttex}

\LLNCS{}{
  \setcopyright{none}
  \copyrightyear{2017}
  \usepackage{fixtheorems}
}

\usepackage{lib}
\usepackage{lrlib}
\usepackage{fixtheorems}

\usepackage{natbib}

\begin{document}

\newcommand{\printernote}{This paper is best printed in color: we use
  $\ufont{\text{blue}}$ for ML and $\lfont{\text{red}}$ for the linear
  language. One of the theorems we prove is that if
  $\uty \ulcompat \lty$ and $\utypr \ulcompat \lty$ then
  $\uty = \utypr$.}

\title{
  \texorpdfstring{Fab\textsf{\ucolor{U}}\textsf{\lcolor{L}}ous}{FabULous}
  Interoperability for ML and a Linear Language%
  \Not{\LLNCS}{\texorpdfstring{\footnote{\printernote}}{}}{}
}

\Anonymous{
  \author{Anonymous for submission}
}{
  \LLNCS{
    \author{Gabriel Scherer\inst{1}\inst{2} \and
            Max New\inst{1} \and
            Nick Rioux\inst{1} \and
            Amal Ahmed\inst{1}\inst{3}}
    \institute{Northeastern University, US
    \and INRIA Saclay, France
    \and INRIA Paris, France}
  }{
    \author{Gabriel Scherer}
    \author{Max New}
    \author{Nick Rioux}
    \author{Amal Ahmed}
    \affiliation{\institution{Northeastern University}\country{US}}
  }
}

\maketitle

\begin{abstract}
  Instead of a monolithic programming language trying to cover all
  features of interest, some programming systems are designed by
  combining together simpler languages that cooperate to cover the
  same feature space. This can improve usability by making each part
  simpler than the whole, but there is a risk of \emph{abstraction
    leaks} from one language to another that would break expectations
  of the users familiar with only one or some of the involved
  languages.

  We propose a formal specification for what it means for a given
  language in a multi-language system to be usable without leaks: it
  should embed into the multi-language in a \emph{fully abstract} way,
  that is, its contextual equivalence should be unchanged in the
  larger system.

  To demonstrate our proposed design principle and formal
  specification criterion, we design a multi-language programming
  system that combines an ML-like statically typed functional language
  and another language with linear types and linear state. Our goal is
  to cover a good part of the expressiveness of languages that mix
  functional programming and linear state (ownership), at only
  a fraction of the complexity.  We prove that the embedding of ML
  into the multi-language system is fully abstract: functional
  programmers should not fear abstraction leaks. We show examples of
  combined programs demonstrating in-place memory updates and safe
  resource handling, and an implementation extending OCaml with our
  linear language.
\end{abstract}

\begin{version}{\VeryShort}
\vspace{-1em}
  \textbf{Note:} Due to severe space restrictions, many details have been omitted
  from this presentation of our work. We strongly encourage the reader
  to consult the complete version at \url{https://arxiv.org/pdf/1707.04984}
\vspace{-1em}
\end{version}

\section{Introduction}

Feature accretion is a common trend among mature but actively evolving
programming languages, including C++, Haskell, Java, OCaml, Python, and
Scala. Each new feature strives for generality and expressiveness, and
may provide a large usability improvement to users of the particular
problem domain or programming style it was designed to empower (e.g., XML
documents, asynchronous communication, staged evaluation).  But 
feature creep in general-purpose languages may also make it harder for
programmers to master the language as a whole, degrade the 
user experience (e.g., leading to more cryptic error messages),
require additional work on the part of tooling providers, and lead to
fragility in 
language implementations.

A natural response to increased language complexity is to define subsets
of the language designed for a better programming experience.  For
instance, a subset can be easier to teach (e.g., ``Core'' 
ML\footnote{\url{https://caml.inria.fr/pub/docs/u3-ocaml/ocaml-ml.html}},  
Haskell 98 as opposed to GHC Haskell, Scala mastery
levels\footnote{\url{http://www.scala-lang.org/old/node/8610}}); it
can facilitate static analysis or decrease the risk of programming 
errors, 
while remaining sufficiently expressive for the target users'
needs (e.g., MISRA C, Spark/Ada); it can enforce a common style within
a company; or it can be designed to encourage a transition to deprecate
some ill-behaved language features (e.g., strict Javascript). 

Once a subset has been selected, it may be the case that users write
whole programs purely in the subset (possibly using tooling to enforce
that property), but programs will commonly rely on other libraries
that are not themselves implemented in the same subset of the
language.  If users stay in the subset while using these libraries,  
they will only interact with the part of the library whose interface is
expressible in the subset.  But does the behavior of the library
respect the expectations of users who only know the subset? 
When calling a function from within the subset breaks subset
expectations, it is a sign of \emph{leaky abstraction}.

How should we design languages with useful subsets that manage
complexity and avoid abstraction leaks?

We propose to look at this question from a different, but equivalent,  
angle: instead of designing a single big monolithic language with some
nicer subsets, we propose to consider \emph{multi-language}
programming systems where several smaller programming languages
interact together to cover the same feature space. Each language or
sub-combination of languages is a subset, in the above sense, of the
multi-language, and there is a clear definition of \emph{abstraction
  leaks} in terms of user experience: a user who only knows some
of the languages of the system should be able to use the 
multi-language system, interacting with code written in the other
languages, without have their expectations violated. If we write a 
program in Java and call a function that, internally, is
implemented in Scala, there should be no surprises---our experience
should be the same as when calling a pure Java function. Similarly,
consider the subset of Haskell that does not contain \texttt{IO}
(input-output as a type-tracked effect): the expectations of a user of
this language, for instance in terms of valid equational reasoning,
should not be violated by adding \texttt{IO} back to the language---in
the absence of the abstraction-leaking \texttt{unsafePerformIO}. 

We propose a \emph{formal specification} for a ``no abstraction
leaks'' guarantee that can be used as a design criterion to design new
multi-language systems, with graceful interoperation properties. It is
based on the formal notion of \emph{full abstraction} which has
previously been used to study the denotational semantics of programming
languages~\citep*{meyer-sieber88,milner77,cartwright92,jeffrey95,abramsky00},
and the formal property of
compilers~\citep*{ahmed08:tccpoe,ahmed11:epcps,devriese16,new16:facue,patrignani15}, 
but not for 
user-facing languages. A 
compiler $C$ from a source language $S$ to a target language $T$ is
\emph{fully abstract} if, whenever two source terms $s_1$ and $s_2$
are indistinguishable in $S$, their translations $C(s_1)$ and $C(s_2)$
are indistinguishable in $T$. In a multi-language $G + E$ formed of a
general-purpose, user-friendly language $G$ and a more advanced
language $E$---one that provides an \emph{e}scape hatch for
\emph{e}xperts to write code that can't be implemented in $G$---we say
that $E$ does not \emph{leak} into $G$ if the embedding of $G$ into the
multi-language $G + E$ is fully abstract. 

To demonstrate that our formal specification is reasonable, we design 
a novel multi-language programming system that satisfies it. Our
multi-language $\ullang$ combines a general-purpose functional
programming language $\ulang$ (unrestricted) of the ML family with
an advanced language $\llang$ (linear) with \emph{linear types} and linear
state. It is less convient to program in $\llang$'s restrictive type
system, but users can write programs in $\llang$ that could not be
written in $\ulang$: they can use linear types, locally, to enforce
resource usage protocols (typestate), and they can use linear state 
and the linear ownership discipline to write programs that do in-place
update to allocate less memory, yet remain observationally pure.

Consider for example the following mixed-language program. The blue
fragments are written in the general-purpose, user-friendly functional
language, while the red fragments are written in the linear
language.  The boundaries $\ufont{UL}$ and $\lfont{LU}$ allow
switching between languages. The program reads all lines from a file,
accumulating them in a list, and concatenating it into a single string
when the end-of-file (EOF) is reached.
\begin{small}
\begin{lstlisting}
/!let concat_lines path : String = UL(/*
  loop (open LU(/!path!/)) LU(/!Nil!/)
  where rec loop handle LU(/!acc : List String!/) =
    match line handle with
    | Next line LU(/!handle!/) -> loop handle LU(/!Cons line acc!/)
    | EOF handle -> close handle; LU(/!rev_concat "\n" acc!/)*/)!/
\end{lstlisting}
\end{small}
The linear type system ensures that the file handle is properly
closed: removing the $\lfont{close}~\lfont{handle}$ call would give a type
error. On the other hand, only the parts concerned with the
resource-handling logic need to be written in the red linear language; 
the user can keep all general-purpose logic (here, how to accumulate
lines and what to do with them at the end) in the more convenient
general-purpose blue language---and call this function from a
blue-language program. Fine-grained boundaries allow users to rely on
each language's strength and to use the advanced features only when
necessary. 

In this example, the file-handle API specifies that the call to
$\lfont{line}$, which reads a line, returns the data at type
$\llumpedty{\ufont{String}}$.  The latter represents how $\usymbol$ 
values of type $\ufont{String}$ can be put into a \emph{lump} type to
be passed to the linear world where they are treated as opaque
blackboxes that must be passed back to the ML world for consumption.
For other examples, such as in-place list manipulation or
transient operations on an persistent data structure, we will need a
deeper form of interoperability where the linear world creates,
dissects or manipulates $\usymbol$ values. To enable this, our
multi-language supports translation of types from one language to the
other, using a \emph{type compatibility} relation $\uty \ulcompat
\lty$ between $\ulang$ types $\uty$ and $\llang$ types $\lty$.


We claim the following contributions:

\begin{enumerate}
\item We propose a formal specification of what it means for advanced
  language features to be introduced in a (multi-)language system
  without introducing a class of abstraction leaks that break
  equational reasoning. This specification captures a useful
  \emph{usability} property, and we hope it will help us and others
  design more usable programming languages, much like the formal
  notion of \emph{principal types} served to better understand and
  design type inference systems.
\item We design a simple linear language, $\llang$, that supports
  linear state (\shortref{sec:u-and-l}). This simple design for linear
  state is a contribution of its own. A nice property of the language
  (shared by some other linear languages) is that the code has both an
  imperative interpretation---with in-place memory update, which
  provides resource guarantees---and a functional
  interpretation---which aids program reasoning. The imperative and
  functional interpretations have different resource usage, but the
  same input/output behavior.
\item We present a multi-language programming system $\ullang$
  combining a core ML language, $\ulang$ ($\usymbol$ for Unrestricted,
  as opposed to Linear) with $\llang$ 
  and prove that the embedding of the ML language $\ulang$ in
  $\ullang$ is fully abstract (\shortref{sec:multi-language}).
  Moreover, the multi-language is designed to ensure that our full
  abstraction result is stable under extension of the embedded ML
  language $\ulang$.   
\begin{version}{\Not\VeryShort}
\item We define a logical relation and prove parametricity for
  $\ullang$. The logical relation illustrates, semantically, why one
  can reason functionally about programs in $\ullang$ despite the
  presence of state and strong updates (\shortref{sec:logrel}).
\item We evaluate the resulting language design by providing examples
  of hybrid $\ullang$ programs that exhibit programming patterns
  inaccessible to ML alone, such that safe in-place updates and
  typestate-like static protocol enforcement
  (\shortref{sec:examples})\True{}{, and describe our prototype implementation
  (\shortref{sec:implementation})}.
\end{version}
\end{enumerate}

\section{The \texorpdfstring{$\ulang$}{U} and \texorpdfstring{$\llang$}{L} languages}
\label{sec:u-and-l}

\begin{figure}
  \begin{sdisplaymath}
    \begin{array}{lcrl}
      \mbox{\textit{Types}} &
      \uty
      & \bnfdef
      & \ualpha
        \bnfalt
        \upairty \utyone \utytwo
        \bnfalt
        \uunitty
        \bnfalt
        \ufunty \utyone \utytwo
        \bnfalt
        \usumty \utyone \utytwo
        \bnfalt
        \umuty \ualpha \uty
        \bnfalt
        \uforallty \ualpha \uty
      \\[3pt]

      \mbox{\textit{Expr.}} &
      \ue & \bnfdef
      & \ux
        \bnfalt 
        \upaire \ueone \uetwo
        \bnfalt
        \uprje 1 \ue
        \bnfalt
        \uprje 2 \ue
        \bnfalt
        \uunite
        \bnfalt
        \uletunite \ueone \uetwo
        \bnfalt
        \ufune \ux \uty \ue
        \bnfalt
        \uappe \ueone \uetwo
        \bnfalt
      \\ & &
      & \usume i \ue
        \bnfalt
        \ucasee \uepr \uxone \ueone \uxtwo \uetwo
        \bnfalt
        \ufolde {\umuty \ualpha \uty} \ue
        \bnfalt
        \uunfolde \ue
        \bnfalt
        \uabstre \ualpha \ue
        \bnfalt
        \uinste \ue \uty
      \\[3pt]

      \mbox{\textit{Values}} &
      \uv & \bnfdef
      & \ux
        \bnfalt
        \upaire \uvone \uvtwo
        \bnfalt
        \uunite
        \bnfalt
        \ufune \ux \uty \ue
        \bnfalt
        \usume 1 \uv \bnfalt \usume 2 \uv
        \bnfalt
        \ufolde {\umuty \ualpha \uty} \uv
        \bnfalt
        \uabstre \ualpha \uv
      \\[3pt]

      \mbox{\textit{Contexts}} & \uGamma & \bnfdef
      & \uempty
        \bnfalt
        \uGamma, \var \ux \uty
        \bnfalt
        \uGamma, \ualpha
    \end{array}
  \end{sdisplaymath}
  \caption{Unrestricted Language: Syntax}
  \label{fig:u:syntax}
\end{figure}

The unrestricted language $\ulang$ is a run-of-the-mill idealized ML
language with functions, pairs, sums, iso-recursive types and
polymorphism. It is presented in its explicitly typed form---we will
not discuss type inference in this work. The full syntax is described
in \shortref{fig:u:syntax}, and the typing rules in
\shortref{fig:u:staticsem}.  The dynamic semantics is completely
standard. Having binary sums, binary products and iso-recursive types
lets us express algebraic datatypes in the usual way.

\begin{figure}
  \fbox{\small$\ujudg{\uGamma}{\ue}{\uty}$}
  \begin{small}
  \begin{mathpar}
    \inferrule
    {\var \ux \uty \in \uGamma}
    {\ujudg \uGamma \ux \uty}

    \inferrule
    { }
    {\ujudg \uGamma \uunite \uunitty}

    \inferrule
    {\ujudg \uGamma \ue \uunitty \+
     \ujudg \uGamma \uepr \uty}
    {\ujudg \uGamma {\uletunite \ue \uepr} \uty}
\\
    \inferrule
    {\ujudg \uGamma  \ueone  \utyone \+
     \ujudg \uGamma  \uetwo  \utytwo}
    {\ujudg \uGamma
            {\upaire \ueone \uetwo}
            {\upairty \utyone \utytwo}}

    \inferrule
    {\ujudg \uGamma \ue {\upairty \utyone \utytwo}}
    {\ujudg \uGamma {\uprje \ix \ue} \utyi}
\\
    \inferrule
    {\ujudg {\uGamma, \var \ux \uty} \ue \utypr}
    {\ujudg \uGamma {\ufune \ux \uty \ue} {\ufunty \uty \utypr}}

    \inferrule
    {\ujudg{\uGamma}{\ue}{\ufunty{\utypr}{\uty}} \+
     \ujudg{\uGamma}{\uepr}{\utypr}}
    {\ujudg{\uGamma}{\uappe{\ue}{\uepr}}{\uty}}
\\
    \inferrule
    {\ujudg \uGamma \ue \utyi}
    {\ujudg \uGamma {\usume \ix \ue} {\usumty \utyone \utytwo}}

    \inferrule
    {\begin{array}{l}\\\ujudg \uGamma \ue {\usumty \utyone \utytwo}\end{array} \+
     \begin{array}{l}
       \ujudg {\uGamma,\uxone:\utyone} \ueone \uty\\
       \ujudg {\uGamma,\uxtwo:\utytwo} \uetwo \uty
     \end{array}}
    {\ujudg \uGamma {\ucasee \ue \uxone \ueone \uxtwo \uetwo} \uty}
\\
    \inferrule
    {\ujudg \uGamma \ue {\uunfoldedty \ualpha \uty}}
    {\ujudg \uGamma {\ufolde {\umuty \ualpha \uty} \ue} {\umuty \ualpha \uty}}

    \inferrule
    {\ujudg \uGamma \ue {\umuty \ualpha \uty}}
    {\ujudg \uGamma  {\uunfolde \ue} {\uunfoldedty \ualpha \uty}}
\\
    \inferrule
    {\ujudg {\uGamma, \ualpha} \uv \uty}
    {\ujudg \uGamma {\uabstre \ualpha \uv} {\uforallty \ualpha \uty}}

    \inferrule
    {\ujudg \uGamma \ue {\uforallty \ualpha \uty}
     \+ \uwf \uGamma \utypr}
    {\ujudg \uGamma {\uinste \ue \utypr} {\subst \uty \utypr \ualpha}}
%
%
  \end{mathpar}
  \end{small}
  \caption{Unrestricted Language: Static Semantics}
  \label{fig:u:staticsem}
\end{figure}

The novelty lies in the linear language $\llang$, which we present in
several steps. As is common in $\lambda$-calculi with references, the
small-step operational semantics is given for a language that is not
exactly the surface language in which programs are written, because
memory allocation returns \emph{locations} $\lloc$ that are not in the
grammar of surface terms. Reductions are defined on
\emph{configurations}, a local store paired with a term in 
a slightly larger \emph{internal} language. We have two type systems,
a type system on surface terms, that does not mention locations and
stores---which is the one a programmer needs to know---and a type
system on configurations, which contains enough static information to
reason about the dynamics of our language and prove subject
reduction. Again, this follows the standard structure of syntactic
soundness proofs for languages with a mutable store.

\begin{version}{\Not\VeryShort}
We present the surface language and type system in
\shortref{subsec:l-intro}, except for the language fragment
manipulating the linear store which is presented in
\shortref{subsec:l:store}. Finally, the internal terms, their
typing and reduction semantics are presented in
\shortref{subsec:l-conf}.
\end{version}

\subsection{The Core of \texorpdfstring{$\llang$}{L}}
\label{subsec:l-intro}

\shortref{fig:l:surface-syntax} presents the surface syntax of our
linear language $\llang$. For the syntactic categories of types
$\lty$, and expressions $\le$, the last line contains the
constructions related to the linear store that we only discuss in
\shortref{subsec:l:store}.


\begin{figure}
  \begin{sdisplaymath}
    \begin{array}{l@{~}c@{~}r@{~}ll}
      \mbox{\textit{Types}} &
      \lty
      & \bnfdef
      & \lopairty \ltyone \ltytwo
        \bnfalt
        \lunitty
        \bnfalt
        \lofunty \ltyone \ltytwo
        \bnfalt
        \losumty \ltyone \ltytwo
        \bnfalt
        \lmuty \lalpha \lty \bnfalt
        \lalpha
        \bnfalt
        \lbangty \lty
        \bnfalt
        \lboxedty \lty \bnfalt \lemptyboxty
      \\[3pt]

      \mbox{\textit{Expr.}} &
      \le & \bnfdef
      & \lx
        \bnfalt  
        \lpaire \leone \letwo
        \bnfalt
        \lletpaire \lvone \lvtwo \leone \letwo
        \bnfalt
        \lunite
        \bnfalt
        \lletunite \leone \letwo
        \bnfalt
        \lfune \lx \lty \le
        \bnfalt
        \lappe \leone \letwo
        \bnfalt
      \\ & &
      & \lsume 1 \le \bnfalt \lsume 2 \le \bnfalt
        \lcasee \lepr \lxone \leone \lxtwo \letwo
        \bnfalt
        \lfolde {\lmuty \lalpha \lty} \le
        \bnfalt
        \lunfolde \le
        \bnfalt
      \\ & &
      & \lsrcsharee \le
        \bnfalt
        \lsrccopye \le
        \bnfalt
        \lnewe \le \bnfalt \lfreee \le
        \bnfalt
        \lboxe \le
        \bnfalt
        \lunboxe \le
      \\[3pt]

      \mbox{\textit{Values}} &
      \lv & \bnfdef
      & \lx
        \bnfalt
        \lpaire \lvone \lvtwo
        \bnfalt
        \lunite
        \bnfalt
        \lfune \lx \lty \le
        \bnfalt
        \lsume 1 \lv \bnfalt \lsume 2 \lv
        \bnfalt
        \lfolde {\lmuty \lalpha \lty} \lv
        \bnfalt
        \lsrcsharee \lv
      \\[3pt]

      %
      %
      \mbox{\textit{Contexts}} & \lGamma & \bnfdef
      & \lemptyGamma
        \bnfalt
        \lGamma, \var \lx \lty
    \end{array}
  \end{sdisplaymath}
  \caption{Linear Language: Surface Syntax}
  \label{fig:l:surface-syntax}
\end{figure}

In technical terms, our linear type system is exactly propositional
intuitionistic linear logic, extended with iso-recursive
types. 
For simplicity and because we did not need them, our current system
also does not have polymorphism or additive/lazy pairs $\lowithty
\ltyone \ltytwo$. Additive pairs would be a trivial addition, but
polymorphism would require more work when we define the multi-language
semantics in \shortref{sec:multi-language}.

In less technical terms, our type system can enforce that values be
used \emph{linearly}, meaning that they cannot be duplicated or
erased, they have to be deconstructed exactly once. Only some types
have this linearity restriction; others allow duplication and sharing
of values at will. We can think of linear values as \emph{resources} to
be spent wisely; for any linear value somewhere in a term, there can
be only one way to access this value, so we can interpret the language as
enforcing an \emph{ownership} discipline where whoever points to
a linear value owns it.

\begin{version}{\VeryShort}
  In particular, linear functions of type $\lofunty \ltyone \ltytwo$
  must be called exactly once, and their results must in turn be
  consumed -- they can safely capture linear resources. On the other
  hand, the non-linear, duplicable values are those at types of the
  form $\lbangty \lty$ --- the \emph{exponential} modality of linear
  logic. If the term $\le$ has duplicable type $\lbangty \lty$, then
  the term $\lsrccopye \le$ has type $\lty$: this creates a local copy
  of the value that is uniquely-owned by its receiver and must be
  consumed linearily.
\end{version}

\begin{version}{\Not\VeryShort}
The types of linear values are the type of linear pairs
$\lopairty \ltyone \ltytwo$, of linear disjoint unions
$\losumty \ltyone \ltytwo$, of linear functions
$\lofunty \ltyone \ltytwo$, and of the linear unit type $\lunitty$.
For example, a linear function must be called exactly once, and its
result must in turn be consumed -- such linear functions can safely
capture linear resources. The expression-formers at these types use
the same syntax as the unrestricted language $\ulang$, with the
exception of linear pair deconstruction
$\lletpaire \lvone \lvtwo \leone \letwo$, which names both members of
the deconstructed pair at once. A linear pair type with projection
would only ever allow to observe one of the two members; this would
correspond to the additive/lazy pairs $\lowithty \ltyone \ltytwo$,
where only one of the two members is ever computed.

The types of non-linear, duplicable values are the types of the form
$\lbangty \lty$---the \emph{exponential} modality of linear logic. If
$\le$ has type $\lty$, the term $\lsrcsharee \le$ has type
$\lbangty \lty$. Values of this type are not uniquely owned, they can
be shared at will. If the term $\le$ has duplicable type
$\lbangty \lty$, then the term $\lsrccopye \le$ has type $\lty$: this
creates a local copy of the value that is uniquely-owned by its
receiver and must be consumed linearily.
\end{version}

\begin{figure}
  \fbox{\small{$\ctxjoin \lGammaone \lGammatwo$}}

  \begin{smathpar}
    \begin{array}{lll@{\qquad\qquad}l}
      \ctxjoin {(\lGammaone, \var \lx {\lbangty \lty})}
               {(\lGammatwo, \var \lx {\lbangty \lty})}
      & \defeq
      & (\ctxjoin \lGammaone \lGammatwo), \var \lx {\lbangty \lty}
      &
      \\ \ctxjoin {(\lGammaone, \var \lx \lty)} \lGammatwo
      & \defeq
      & (\ctxjoin \lGammaone \lGammatwo), \var \lx \lty
      & (\lx \notin \lGammatwo)
      \\ \ctxjoin \lGammaone {(\lGammatwo, \var \lx \lty)}
      & \defeq
      & (\ctxjoin \lGammaone \lGammatwo), \var \lx \lty
      & (\lx \notin \lGammaone)
    \end{array}
  \end{smathpar}

  \fbox{\small$\ljudgnoconf \lGamma \le \lty$}
  \begin{smathpar}
    \inferrule
    { }
    {\ljudgnoconf {\lbangty \lGamma, \var \lx \lty} \lx \lty}

    \inferrule
    {\ljudgnoconf \lGammaone \leone \ltyone \+
     \ljudgnoconf \lGammatwo \letwo \ltytwo}
    {\ljudgnoconf
      {\ctxjoin \lGammaone \lGammatwo}
      {\lpaire \leone \letwo}
      {\lopairty \ltyone \ltytwo}}

    \inferrule
    {\begin{array}{c}
       \ljudgnoconf \lGamma \le {\lopairty \ltyone \ltytwo} \\
       \ljudgnoconf {\lGammapr, \var \lxone \ltyone, \var \lxtwo \ltytwo}
         \lepr \lty
     \end{array}}
    {\ljudgnoconf
      {\ctxjoin \lGamma \lGammapr}
      {\lletpaire \lxone \lxtwo \le \lepr}
      \lty}
\\
    \inferrule
    { }
    {\ljudgnoconf {\lbangty \lGamma} \lunite \lunitty}

    \inferrule
    {\ljudgnoconf \lGamma \le \lunitty \+
     \ljudgnoconf \lGammapr \lepr \lty}
    {\ljudgnoconf
      {\ctxjoin \lGamma \lGammapr}
      {\lletunite \le \lepr}
      \lty}

    \inferrule
    {\ljudgnoconf {\lGamma, \var \lx \lty} \le \ltypr}
    {\ljudgnoconf \lGamma {\lfune \lx \lty \le} {\lofunty \lty \ltypr}}

    \inferrule
    {\ljudgnoconf \lGamma \le {\lofunty{\ltypr}{\lty}} \+
     \ljudgnoconf \lGammapr \lepr \ltypr}
    {\ljudgnoconf
      {\ctxjoin \lGamma \lGammapr}
      {\lappe \le \lepr}
      \lty}

    \inferrule
    {\ljudgnoconf \lGamma \le \ltyi}
    {\ljudgnoconf \lGamma
      {\lsume \ix \le} {\losumty \ltyone \ltytwo}}

    \inferrule
    {\begin{array}{c}\\
       \ljudgnoconf \lGamma \le {\losumty \ltyone \ltytwo}
     \end{array} 
     \+
     \begin{array}{c}
     \ljudgnoconf {\lGammapr,\lxone:\ltyone} \leone \lty \\
     \ljudgnoconf {\lGammapr,\lxtwo:\ltytwo} \letwo \lty
     \end{array}}
    {\ljudgnoconf
      {\ctxjoin \lGamma \lGammapr}
      {\lcasee \le \lxone \leone \lxtwo \letwo}
      \lty}

    \inferrule
    {\ljudgnoconf {\lbangty \lGamma} \le \lty}
    {\ljudgnoconf
      {\lbangty \lGamma}
      {\lsrcsharee \le}
      {\lbangty \lty}}

    \inferrule
    {\ljudgnoconf \lGamma \le {\lbangty \lty}}
    {\ljudgnoconf \lGamma {\lsrccopye \le} \lty}
\\
    \bidirliniso
      {\lmuty \lalpha \lty}
      {\lunfoldedty \lalpha \lty}
      {\lunfolde {}}
      {\lfolde {\lmuty \lalpha \lty} {}}

\bidirliniso
  {\lunitty}
  {\lemptyboxty}
  {\lnewe {}}
  {\lfreee {}}

\bidirliniso
  {\lboxedty \lty}
  {\lunboxedty \lty}
  {\lunboxe {}}
  {\lboxe {}}  \end{smathpar}
  \caption{Linear Language: Surface Static Semantics}
  \label{fig:l:surface-static-sem}
  \label{fig:l:store-surface-static-sem}
\end{figure}
\begin{version}{\VeryShort}
\begin{figure}
  head reduction
  \quad
  \fbox{\small{$\le \lheadredexstep \lepr$}}
  \quad
  \fbox{\small{$\hconf \lstore \le \lheadredexstep \hconf \lstorepr \lepr$}}
  \begin{smathpar}
    \bidirlheadredexstep
      {\hconf \lemptystore \lunite}
      {\hconf {\stsingleton \lloc \lemptyloc} \lloc}
      {\lnewe {}}
      {\lfreee {}}

    \bidirlheadredexstep
      {\hconf
        {\stextempty \lstore \lloc}
        {\lpaire \lloc \lv}}
      {\hconf
        {\stsingletonconf \lloc \lstore \lv}
        \lloc}
      {\lboxe {}}
      {\lunboxe {}}

    \lappe {(\lfune \lx \lty \le)} \lv
    \;\lheadredexstep\;
    {\subst \le \lv \lx}

    \lsrccopye {(\lsharee \lstore \lstorety {\lsume \ix \lv})}
    \;\lheadredexstep\;
    \lsume \ix {\lsrccopye {(\lsharee \lstore \lstorety \lv)}}

    \hconf \lemptystore
      {\lsrccopye {(\lsharee \lstore \lstorety {\lfune \lx \lty \le})}}
    \;\lheadredexstep\;
    \hconf \lstore {\lfune \lx \lty \le}

    {\lsrccopye
      {(\lsharee {\stsingleton \lloc \lemptyloc} {\ladead \lloc \lty} \lloc)}}
    \;\lheadredexstep\;
    \lnewe \lunite

    \begin{array}{r}
      {\lsrccopye
        {(\lsharee
          {\stsingletonconf \lloc \lstore \lv}
          {\lalive {\lbangty \lGamma} \lstorety \lloc \lty}
        \lloc)}}
        \\
      \lheadredexstep\;
      \lboxe {\lpaire 
        {\lnewe \lunite}
        {\lsrccopye {(\lsharee \lstore \lstorety \lv)}}}
    \end{array}

  \end{smathpar}
  \caption{Internal linear language: Typing and reduction (excerpt)}
  \label{fig:l:internal-reduction-excerpt}
\end{figure}
\end{version}
This resource-usage discipline is enforced by the surface typing rules
of $\llang$, presented in \shortref{fig:l:surface-static-sem}. They are
exactly the standard (two-sided) logical rules of intuitionistic
linear logic, annotated with program terms. The non-duplicability of
linear values is enforced by the way contexts are merged by the
inference rules: if $\leone$ is type-checked in the context
$\lGammaone$ and $\letwo$ in $\lGammatwo$, then the linear pair
$\lpaire \leone \letwo$ is only valid in the combined context
$\ctxjoin \lGammaone \lGammatwo$. The $(\ctxjoin{}{})$ operation is
partial; this combined context is defined only if the variables shared
by $\lGammaone$ and $\lGammatwo$ are duplicable---their type is of
the form $\lbangty \lty$. In other words, a variable at
a non-duplicable type in $\ctxjoin \lGammaone \lGammatwo$ cannot
possibly appear in both $\lGammaone$ and $\lGammatwo$: it must appear
exactly once\footnote{Standard presentations of linear logic force
  contexts to be completely distinct, but have a separate rule to
  duplicate linear variables, which is less natural for
  programming.}.

\begin{version}{\Not\VeryShort}
A good way to think of the linear judgment
$\ljudgnoconf \lGamma \le \lty$ is that the evaluation of $\le$
\emph{consumes} the linear variables of $\lGamma$; it is thus natural
that the strict pair $\lpaire \leone \letwo$ would need separate sets
of resources $\lGammaone$ and $\lGammatwo$, as it evaluates both
members to return a value. On the other hand, case elimination
$\lcasee \le \lxone \leone \lxtwo \letwo$ reuses the same context
$\lGammapr$ in both branches $\leone$ and $\letwo$: only one will be
evaluated, so they do not compete for resources.
\end{version}

\begin{version}{\VeryShort}
  The expression $\lsrcsharee \le$ takes a term at some type $\lty$
  and creates a ``shared'' term, whose value will be duplicable. Its
  typing rule uses a context of the form $\lbangty \lGamma$, which is
  defined as the pointwise application of the $(\lbangty)$ connectives
  to all the types in $\lGamma$. In other words, the context of this
  rule must only have duplicable types : a term can only be made
  duplicable if it does not depend on linear resources from the
  context. Otherwise, duplicating the shared value could break the
  unique-ownership discipline on these linear resources.
\end{version}

\begin{version}{\Not\VeryShort}
The variable rule does not expect a context of the form
$\lGamma, \var \lx \lty$ but of the form
$\lbangty \lGamma, \var \lx \lty$. 
Here $\lbangty \lGamma$ is a notation
for the pointwise application of the $(\lbangty)$ connective to all
the types in $\lGamma$---i.e., all types in $\lbangty \lGamma$ are of
the form $\lbangty \lty$. 
This means that the variable rule can only be used
when all variables in the context are duplicable, except maybe the
variable that is being used. A context of the form
$\lGamma, \var \lx \lty$ would allow us to forget some variable present
in the context; in our judgment $\ljudgnoconf \lGamma \le \lty$, all
non-duplicable variables in $\lGamma$ must appear (once) in $\le$.

The form $\lbangty \lGamma$ is also used in the typing rule for
$\lsrcsharee \le$: a term can only be made duplicable if it does not
depend on linear resources from the context. Otherwise, duplicating
the shared value could break the unique-ownership discipline on these
linear resources.
\end{version}

Finally, the linear isomorphism notation for $\lfolde {}$ and
$\lunfolde {}$ in \shortref{fig:l:surface-static-sem} defines them as
primitive functions, at the given linear function type, in the empty
context -- using them does not consume resources. This notation also
means that, operationally, these two operations shall be inverses of
each other. The rules for the linear store type $\lboxedty \lty$ and
$\lemptyboxty$ are described in \shortref{subsec:l:store}.

\begin{version}{\Not\VeryShort}
\begin{lemma}[Context joining properties]
  \label{lem:l:join-assoc-commut}
  Context joining $(\ctxjoin{}{})$ is partial but associative and
  commutative. In particular, if
  $\ctxjoin {(\ctxjoin \lGammaone \lGammatwo)} \lGammapr$ is defined,
  then both $\ctxjoin \lGammai \lGammapr$ are defined.
\end{lemma}
\end{version}

\subsection{Linear Memory in \texorpdfstring{$\llang$}{L}}
\label{subsec:l:store}

The surface typing rules for the linear store are given at the end of
\shortref{fig:l:store-surface-static-sem}. The linear type
$\lboxedty \lty$ represents a memory location that holds a value of
type $\lty$. The type $\lemptyboxty$ represents a location that
has been allocated, but does not currently hold a value.
The primitive
operations to act on this type are given as linear isomorphisms:
$\lnewe {}$ allocates, turning a unit value into an empty location;
conversely, $\lfreee {}$ reclaims an empty location. Putting a value
into the location and taking it out are expressed by $\lboxe {}$ and
$\lunboxe {}$, which convert between a pair of an empty location and
a value, of type $\lunboxedty \lty$, and a full location, of type
$\lboxedty \lty$.

For example, the following program takes a full reference and a value,
and swaps the value with the content of the reference:

\LLNCS{
  \begin{sdisplaymath}
    \lfune {\lmetavar{p}{}{}} {\lopairty {(\lboxedty \lty)} \lty}
    {\lletpaire
        {\lmetavar{r}{}{}} \lx
        {\lmetavar{p}{}{}}
        {\lletpaire
        {\lmetavar{l}{}{}} {\lmetavar{x}{l}{}}
        {\lunboxe {\lmetavar{r}{}{}}}
        {\lpaire {\lboxe {\lpaire {\lmetavar{l}{}{}} \lx}} {\lmetavar{x}{l}{}}}}}
  \end{sdisplaymath}
}{
\begin{sdisplaymath}
  \begin{array}{l}
  \lfune {\lmetavar{p}{}{}} {\lopairty {(\lboxedty \lty)} \lty}{}
  \begin{stackTL}
      \lletpaire
        {\lmetavar{r}{}{}} \lx
        {\lmetavar{p}{}{}} {} \\
      \lletpaire
        {\lmetavar{l}{}{}} {\lmetavar{x}{l}{}}
        {\lunboxe {\lmetavar{r}{}{}}} {} \\
       \lpaire {\lboxe {\lpaire {\lmetavar{l}{}{}} \lx}} {\lmetavar{x}{l}{}}
\end{stackTL}
\end{array}
\end{sdisplaymath}
}
The programming style following from this presentation of linear
memory is functional, or applicative, rather than imperative. Rather
than insisting on the mutability of references---which is allowed by
the linear discipline---we may think of the type $\lboxedty \lty$ as
representing the indirection through the heap that is implicit in
functional programs. In a sense, we are not writing imperative
programs with a mutable store, but rather making explicit the
allocations and dereferences happening in higher-level purely
functional language. In this view, empty cells allow memory reuse.

This view that $\lboxedty \lty$ represents indirection through the
memory suggests we can encode lists of values of type $\lty$ by the
type
$
\lapptyop {LinList} \lty
\defeq
\lmuty \lalpha
  {\losumty
    \lunitty
    {\lboxedty {(\lopairty \lty \lalpha)}}}
$.  The placement of the box inside the sum mirrors
the fact that empty list is represented as an immediate value in
functional languages. From this type definition, one can write
an in-place reverse function on lists of $\lty$ as follows:

\begin{smathpar}
  \begin{array}{c}
    \lappop {fix}
    {\lfune
    {\lfont{rev\_into}}
    {\lofunty {\lapptyop {LinList} \lty}
      {\lofunty {\lapptyop {LinList} \lty} {\lapptyop {LinList} \lty}}}
    {}}\\
    {\lfune
      {\lmetavar{xs}{}{}}
      {\lapptyop {LinList} \lty}
      {\lfune
        {\lmetavar{acc}{}{}}
        {\lapptyop {LinList} \lty} {}}}\\
    \lbcasee {\lunfolde {\lmetavar{xs}{}{}}}
    \ly {(\lletunite \ly {\lmetavar{acc}{}{}})}
    \ly {
  \begin{stackTL}
      \lletpaire {\lmetavar{l}{}{}} {\lmetavar{p}{}{}} {\lunboxe \ly} {} \\
      \lletpaire {\lmetavar{xs}{}{}} \lx {\lmetavar{p}{}{}} {} \\
      \lappe
        {\lappop {rev\_into} {\lmetavar{xs}{}{}}}
        {(\lfolde {} {(\lsume 2 {(\lboxe {\lpaire {\lmetavar{l}{}{}} {\lpaire \lx {\lmetavar{acc}{}{}}}})})})}
    \end{stackTL}}
  \end{array}
\end{smathpar}

\begin{version}{\Not\VeryShort}
This definition uses a fixpoint operator $\lfont{fix}$ that can be
defined, in the standard way, using the iso-recursive type
$\lmuty \lalpha {\lofunty \lalpha {\lofunty \lty \ltypr}}$ of the
strict fixpoint combinator on functions $\lofunty \lty \ltypr$.
\end{version}

Our linear language $\llang$ is a formal language that is not
terribly convenient to program directly. We will not present a full
surface language in this work, but one could easily define syntactic
sugar to write the exact same function as follows:
\begin{smathpar}
  \begin{array}{lllllll}
    \lfont{rev\_into}
    & \lfont{Nil}
    & \lfont{acc}
    & =
    & \lmetavar{acc}{}{}
    &
    &
    \\
    \lfont{rev\_into}
    & (\lappop {Cons} {\lwithe {\lfont{l}} {\lpaire \lx {\lfont{xs}}}})
    & \lmetavar{acc}{}{}
    & =
    & \lfont{rev\_into}
    & {\lfont{xs}}
    & {(\lappop {Cons} {\lwithe {\lfont{l}} {\lpaire \lx {\lfont{acc}}}})}
  \end{array}
\end{smathpar}

One can read this function as the usual functional
\texttt{rev\_append} function on lists, annotated with memory reuse
information: if we assume we are the unique owner of the input list
and won't need it anymore, we can reuse the memory of its cons cells
(given in this example the name $\lfont{l}$) to store the reversed
list. On the other hand, if you read the $\lboxe{}$ and $\lunboxe{}$
as imperative operations, this code expresses the usual imperative
pointer-reversal algorithm.

This double view of linear state occurs in other programming systems
with linear state. It was recently emphasized in \citet*{cogent},
where the functional point of view is seen as easing formal
verification, while the imperative view is used as a compilation
technique to produce efficient C code from linear programs.

\subsection{Internal \texorpdfstring{$\llang$}{L} Syntax and Typing}
\label{subsec:l-conf}

To give a dynamic semantics for $\llang$ and prove it sound, we need
to extend the language with explicit stores and store
locations. Indeed, the allocating term $\lnewe \lunite$ should reduce
to a ``fresh location'' $\lloc$ allocated in some store $\lstore$, and
neither are part of the surface-language syntax. The corresponding
internal typing judgment is more complex, but note that users do not
need to know about it to reason about correctness of surface
programs. The internal typing is essential for the soundness proof,
but also useful for defining the multi-language semantics in
\shortref{sec:multi-language}.

\newcommand{\SectionInternalLL}{\begingroup
\begin{figure}
  \begin{smathpar}
    \begin{array}{l@{~}c@{~}r@{~}ll}
      \mbox{\textit{Types}}
      & \lty
      &
      & \text{(unchanged from \shortref{fig:l:surface-static-sem})}
      \\[3pt]

      \mbox{\textit{Expressions}}
      & \le
      & \bnfadd 
      & \dots
        \bnfalt \lloc
        \bnfalt {\lsharee \lstore \lstorety \le}
      \\[1pt] & &
      & \text{with~} {\lsrcsharee \le \defeq
                      \lsharee \lemptystore \lemptystorety \le}
      \\[3pt]

      \mbox{\textit{Values}}
      & \lv
      & \bnfadd
      & \dots
        \bnfalt
        \lloc
        \bnfalt
        \lsharee \lstore \lstorety \lv
    \end{array}

    \begin{array}{l@{~}c@{~}r@{~}ll}
      \mbox{\textit{Store}} & \lstore & \bnfdef
      & \lemptystore
        \bnfalt
        \stextconf \lstore \lloc \lstore \lv
        \bnfalt
        \stextempty \lstore \lloc
      \\[3pt]
      
      \mbox{\textit{Configurations}} & & \bnfdef & \hconf \lstore \le
      \\[3pt]
            
      \mbox{\textit{Store typing}} & \lstorety & \bnfdef
      & \lemptystorety
        \bnfalt
        \lstorety, \ladead \lloc \lty
      \\ & & &

        \;\; \bnfalt
        \lstorety, \lalive \lGamma \lstoretypr \lloc \lty
    \end{array}
  \end{smathpar}

  \fbox{\small$\lstorejoin \lstoreone \lstoretwo$}
  \quad\text{Union of stores on disjoint locations}
  \quad
  \fbox{\small$\storectxjoin \lstoretyone \lstoretytwo$}
  \quad\text{Union of store typings on disjoint locations}

  \fbox{\small$\ljudg \lstorety \lGamma \lstore \le \lty$}
  \hfill
  \fbox{\small$\ljudgnoconf \lGamma \le \lty
    ~\defeq~
    \ljudg \lemptystorety \lGamma \lemptystore \le \lty$}
  \begin{smathpar}
    \inferrule
    {\ljudg \lstoretyone \lGammaone \lstoreone \leone \ltyone \+
     \ljudg \lstoretytwo \lGammatwo \lstoretwo \letwo \ltytwo}
    {\ljudg {\storectxjoin \lstoretyone \lstoretytwo}
            {\ctxjoin \lGammaone \lGammatwo}
            {\lstorejoin \lstoreone \lstoretwo}
            {\lpaire \leone \letwo}
            {\lopairty \ltyone \ltytwo}}

    \inferrule
    {\ljudg \lstorety \lGamma \lstore \le {\lopairty \ltyone \ltytwo} \+
     \ljudg \lstoretypr {\lGammapr, \var \lxone \ltyone, \var \lxtwo \ltytwo}
            \lstorepr \lepr \lty}
    {\ljudg {\storectxjoin \lstorety \lstoretypr}
            {\ctxjoin \lGamma \lGammapr}
            {\lstorejoin \lstore \lstorepr}
            {\lletpaire \lxone \lxtwo \le \lepr}
            \lty}
\\
    \inferrule
    { }
    {\ljudg \lemptystorety {\lbangty \lGamma, \var \lx \lty} \lemptystore \lx \lty}

    \inferrule
    { }
    {\ljudg \lemptystorety {\lbangty \lGamma} \lemptystore \lunite \lunitty}

    \inferrule
    {\ljudg \lstorety \lGamma \lstore \le \lunitty \+
     \ljudg \lstoretypr \lGammapr \lstorepr \lepr \lty}
    {\ljudg
      {\storectxjoin \lstorety \lstoretypr}
      {\ctxjoin \lGamma \lGammapr}
      {\lstorejoin \lstore \lstorepr}
      {\lletunite \le \lepr}
      \lty}
\\
    \inferrule
    {\ljudg  \lstorety {\lGamma, \var \lx \lty}  \lstore \le \ltypr}
    {\ljudg \lstorety \lGamma
      \lstore {\lfune \lx \lty \le}
      {\lofunty \lty \ltypr}}

    \inferrule
    {\ljudg \lstorety \lGamma \lstore \le {\lofunty{\ltypr}{\lty}} \+
     \ljudg \lstoretypr \lGammapr \lstorepr \lepr \ltypr}
    {\ljudg
      {\storectxjoin \lstorety \lstoretypr}
      {\ctxjoin \lGamma \lGammapr}
      {\lstorejoin \lstore \lstorepr}
      {\lappe \le \lepr}
      \lty}
\\
    \inferrule
    {\ljudg \lstorety \lGamma \lstore \le \ltyi}
    {\ljudg \lstorety \lGamma \lstore
      {\lsume \ix \le} {\losumty \ltyone \ltytwo}}

    \inferrule
    {\ljudg \lstorety \lGamma \lstore \le {\losumty \ltyone \ltytwo}  \+
     \ljudg \lstoretypr {\lGammapr,\lxone:\ltyone} \lstorepr \leone \lty \+
     \ljudg \lstoretypr{\lGammapr,\lxtwo:\ltytwo} \lstorepr \letwo \lty}
    {\ljudg
      {\storectxjoin \lstorety \lstoretypr}
      {\ctxjoin \lGamma \lGammapr}
      {\lstorejoin \lstore \lstorepr}
      {\lcasee \le \lxone \leone \lxtwo \letwo}
      \lty}
\\
    \inferrule
    {\ljudg \lstorety {\lbangty \lGamma} \lstore \le \lty}
    {\ljudg \lemptystorety
      {\lbangty \lGamma}
      \lemptystore
      {\lsharee \lstore \lstorety \le}
      {\lbangty \lty}}

    \inferrule
    {\ljudg \lstorety \lGamma \lstore \le {\lbangty \lty}}
    {\ljudg \lstorety \lGamma \lstore {\lsrccopye \le} \lty}
\\
    \inferrule
    { }
    {\ljudg
      {\ladead \lloc \lty}
      {\lbangty \lGamma}
      {\stsingleton \lloc \lemptyloc}
      \lloc
      {\lemptyboxty}}

    \inferrule
    {\ljudg \lstorety \lGamma \lstore \lv \lty}
    {\ljudg
      {(\lstoreassert \lGamma \lstorety \lloc \lboxedty \lty)}
      {\ctxjoin \lGamma {\lbangty \lGammapr}}
      {\stsingletonconf \lloc \lstore \lv} \lloc {\lboxedty \lty}}
\\
\bidirliniso
  {\lmuty \lalpha \lty}
  {\lunfoldedty \lalpha \lty}
  {\lunfolde {}}
  {\lfolde {\lmuty \lalpha \lty} {}}

\bidirliniso
  {\lunitty}
  {\lemptyboxty}
  {\lnewe {}}
  {\lfreee {}}

\bidirliniso
  {\lboxedty \lty}
  {\lunboxedty \lty}
  {\lunboxe {}}
  {\lboxe {}}
  \end{smathpar}
  \caption{Linear Language: Internal Static Semantics}
  \label{fig:l:internal-syntax}
  \label{fig:l:internal-static-sem}
\end{figure}

The syntax of internal terms and the internal type system are
presented in \shortref{fig:l:internal-static-sem}. Reduction will be
defined on \emph{configurations} $\hconf \lstore \le$, which are pairs
of a store $\lstore$ and a term $\le$. Stores $\lstore$ map
\emph{locations} $\lloc$ to either nothing (the location is empty),
written $\stextempty {} \lloc$, or a value paired with its own local store,
noted $\stextconf {} \lloc \lstore \lv$. Having local stores in this
way, instead of a single global store as is typical in formalizations
of ML, directly expresses the idea of ``memory ownership'' in the
syntax: a term $\le$ ``owns'' the locations that appear in it, and
a configuration $\hconf \lstore \le$ is only well-typed if the domain
of $\lstore$ is exactly those locations. Each store slot, in turn, may
contain a value and the local store owned by the value; in particular,
passing a full location of type $\lboxedty \lty$ transfers ownership
of the location, but also of the store fragment captured by the
value.

Our internal typing judgment $\ljudg \lstorety \lGamma \lstore \le \lty$
checks configurations, not just terms, and relies not only on a typing
context for variables $\lGamma$ but also on a \emph{store typing}
$\lstorety$, which maps the locations of the configuration to typing
assumptions of two forms: $\ladead \lloc \lty$ indicates that $\lloc$
must be empty in the configuration, and
$\lalive \lstorety \lGamma \lloc \lty$ indicates that $\lloc$ is full,
and that the value it contains owns a local store of type $\lstore$
and the resources in $\lGamma$.

Just as linear variables must occur exactly once in a term, locations
have linear types and thus occur exactly once in a term. Our typing
judgment uses disjoint store typings $\storectxjoin \lstoretyone
\lstoretytwo$ to enforce this linearity. Similarly, leaf rules such as
the variable, unit, and location rules enforce that both the store
typing and the store be empty, which enforces that all locations are
used in the term. 

Locations $\lloc$ are always linear, never duplicable. To allow
sharing terms that contain locations, the internal language uses the
internal construction $\lsharee \lstore \lstorety \le$, that
\emph{captures} a local store $\lstore : \lstorety$. This notation is
a binding construct: the locations in $\lstore$ are bound by this
shared term, and not visible outside this term. In particular, the
typing rule for $\lsharee \lstore \lstorety \le$ checks the term $\le$
in the store $\lstore$, but it is itself only valid paired with an
empty store, under the empty store typing. When new copies of a shared
term are made, the local store is copied as well: this is necessary to
guarantee that locations remain linear---and for correctness of
linear state update.

The typing rule for functions $\lfune \lx \lty \le$ lets function
bodies use an arbitrary store typing $\lstorety$. This would be
unsound if our functions were duplicable, but it is a natural and
expressive choice for linear, one-shot functions. To make a function
duplicable, one can share it at type
$\lbangty {(\lofunty \lty \ltypr)}$, whose values are of the canonical
form $\lsharee \lstore \lstorety {\lfune \lx \lty \le}$. It is the
sharing construct, not the function itself, that closes over the local
store.

With the macro-expansion
$\lsrcsharee \le \defeq \lsharee \lemptystore \lemptystorety \le$, any
term $\le$ of the surface language (\shortref{fig:l:surface-syntax}) can
be seen as a term of the internal language
(\shortref{fig:l:internal-syntax}). In particular, we can prove that
the surface and internal typing judgments coincide on surface terms.

\begin{lemma}
  If $\le$ is a surface term of $\llang$, then the surface judgment
  $\ljudgnoconf \lGamma \le \lty$ holds if and only if the internal
  judgment $\ljudg \lemptystorety \lGamma \lemptystore \le \lty$
  holds.
\end{lemma}

The following technical results are used in the soundness proof for
the language -- the subject-reduction result.

\begin{lemma}[Inversion principle for $\llang$ values]
  \label{lem:l:inversion}
  In any complete derivation of
  $\ljudg \lstorety \lGamma \lstore \lv \lty$,
  either $\lv$ is a variable $\lx$,
  or the derivation starts with the introduction rule for $\lty$.
\end{lemma}

For example, if we have
$\ljudg \lstorety \lGamma \lstore \lv {\lbangty \lty}$, then we know
that $\lv$ is either a variable or of the form
$\lsharee \lstorepr \lstoretypr \lvpr$ for some $\lvpr$, but also that
$\lstore = \lemptystore$, $\lstorety = \lemptystorety$ and that $\lGamma$ is
of the form $\lbangty \lGammapr$ for some $\lGammapr$. The latter is
immediate if $\lv$ is $\lsharee \lstorepr \lstoretypr \lvpr$, and also
holds if $\lv$ is a variable.

\begin{lemma}[Weakening by duplicable contexts]
  \label{lem:l:weakening}
  $\ljudg \lstorety \lGammapr \lstore \le \lty$ implies
  $\ljudg \lstorety {\lbangty \lGamma, \lGammapr} \lstore \le \lty$.
\end{lemma}
\endgroup}

\VeryShort{
We work with \emph{configurations} $\hconf \lstore \le$, which are
pairs of a store $\lstore$ and a term $\le$.  Our internal typing
judgment $\ljudg \lstorety \lGamma \lstore \le \lty$ checks
configurations, not just terms, and relies not only on a typing
context for variables $\lGamma$ but also on a \emph{store typing}
$\lstorety$, which maps the locations of the configuration to typing
assumptions.

Unfortunately, due to space limits, we will not present this part of
the type system -- which is not directly exposed to users of the
language.
\Appendices{
  We include it for convenience in
  \fullref{appendix:internal-linear-language}
}{
  See some examples of reduction rules in
  \fullref{fig:l:internal-reduction-excerpt}, and the long version of
  this work.
}
}{\SectionInternalLL}

\subsection{Reduction of Internal Terms}

\newcommand{\SectionInternalLLReduction}{\begingroup
\begin{figure}
  head reduction
  \quad
  \fbox{\small{$\le \lheadredexstep \lepr$}}
  \quad
  \fbox{\small{$\hconf \lstore \le \lheadredexstep \hconf \lstorepr \lepr$}}
  \begin{smathpar}
    \begin{array}{r@{\quad\lheadredexstep\quad}l}
      {\lletpaire \lxone \lxtwo {\lpaire \lvone \lvtwo} \le}
      & {\subst {\subst \le \lvone \lxone} \lvtwo \lxtwo}
      \\
      {\lletunite \lunite \le}
      & \le
      \\
      {\lappe {(\lfune \lx \lty \le)} \lv}
      & {\subst \le \lv \lx}
      \\
      {\lcasee {(\lsume \ix \lv)} \lxone \leone \lxtwo \letwo}
      & {\subst \lei \lv \lxi}
      \\
      {\lunfolde {(\lfolde {\lmuty \lalpha \lty} \lv)}}
      & \lv
    \end{array}

    \infer
    {\le \lheadredexstep \lepr}
    {\hconf \lstore \le \lheadredexstep \hconf \lstore \lepr}
\\
    \bidirlheadredexstep
      {\hconf \lemptystore \lunite}
      {\hconf {\stsingleton \lloc \lemptyloc} \lloc}
      {\lnewe {}}
      {\lfreee {}}

    \bidirlheadredexstep
      {\hconf
        {\stextempty \lstore \lloc}
        {\lpaire \lloc \lv}}
      {\hconf
        {\stsingletonconf \lloc \lstore \lv}
        \lloc}
      {\lboxe {}}
      {\lunboxe {}}
    \\
    \begin{array}{l}
      \lsrccopye {(\lsharee
         {\lstorejoin \lstoreone \lstoretwo}
         {\storectxjoin \lstoretyone \lstoretytwo}
         {\lpaire \lvone \lvtwo})}
      \\
      \lheadredexstep
      \qquad \text{if } \locs\lstorei = \locs\lstoretyi = \locs\lvi
      \\ \lpaire
          {\lsrccopye {\lsharee \lstoreone \lstoretyone \lvone}}
          {\lsrccopye {\lsharee \lstoretwo \lstoretytwo \lvtwo}}
    \end{array}

    \begin{array}{l@{\;\lheadredexstep\;}l}
      \lsrccopye {(\lsharee \lemptystore \lemptystorety \lunite)}
      & \lunite \\
      
      \lsrccopye {(\lsharee \lstore \lstorety {\lsume \ix \lv})}
      & \lsume \ix {\lsrccopye {(\lsharee \lstore \lstorety \lv)}} \\

      \lsrccopye {(\lsharee \lstore \lstorety {\lfolde {} \lv})}
      & \lfolde {}
          {(\lsrccopye {(\lsharee \lstore \lstorety \lv)})} \\
    \end{array}

    \hconf \lemptystore
      {\lsrccopye {(\lsharee \lstore \lstorety {\lfune \lx \lty \le})}}
    \;\lheadredexstep\;
    \hconf \lstore {\lfune \lx \lty \le}

    \lsrccopye
      {(\lsharee \lemptystore \lemptystorety
        {(\lsharee \lstore \lstorety \lv)})}
    \;\lheadredexstep\;
    \lsharee \lstore \lstorety \lv

    {\lsrccopye
      {(\lsharee {\stsingleton \lloc \lemptyloc} {\ladead \lloc \lty} \lloc)}}
    \;\lheadredexstep\;
    \lnewe \lunite

    \begin{array}{r}
      {\lsrccopye
        {(\lsharee
          {\stsingletonconf \lloc \lstore \lv}
          {\lalive {\lbangty \lGamma} \lstorety \lloc \lty}
        \lloc)}}
        \\
      \lheadredexstep\;
      \lboxe {\lpaire 
        {\lnewe \lunite}
        {\lsrccopye {(\lsharee \lstore \lstorety \lv)}}}
    \end{array}
\end{smathpar}

    linear reduction contexts
    \fbox{\small{$
        \ljudg \lstorety \lGamma
          \lstore {\lectxt \hw {\var \square \lty}} \ltypr$}}

    \begin{smathpar}
    \begin{array}{c@{~}c@{~}l}
      \lectxt & \bnfdef
      & \var \square \lty
        \bnfalt
        \lpaire \lectxt \letwo
        \bnfalt
        \lpaire \lv \lectxt
        \bnfalt
        \lletpaire \lvone \lvtwo \lectxt \letwo
        \bnfalt
      \\ &
      & \lletunite \lectxt \le
        \bnfalt
        \lappe \lectxt \le
        \bnfalt
        \lappe \lv \lectxt
        \bnfalt
        \lsrccopye \lectxt
        \bnfalt
      \\ &
      & \lsume 1 \lectxt \bnfalt \lsume 2 \lectxt
        \bnfalt
        \lcasee \lectxt \lxone \leone \lxtwo \letwo
        \bnfalt
      \\ &
      & \lfolde {\lmuty \lalpha \lty} \lectxt
        \bnfalt
        \lunfolde \lectxt
        \bnfalt
      \\ &
      & \lnewe \lectxt \bnfalt \lfreee \lectxt
        \bnfalt
        \lboxe \lectxt
        \bnfalt
        \lunboxe \lectxt
    \end{array}

    \text{typing rules of terms, plus:}\quad
    \infer*{ }
    {\ljudg \lemptystorety \lemptyGamma \lemptystore {(\var \square \lty)} \lty}
    \end{smathpar}

    reduction \fbox{\small{$\hconf \lstore \le \lredexstep \hconf \lstorepr \lepr$}}

\begin{smathpar}
    \inferrule
    {\hconf \lstore \le \lheadredexstep \hconf \lstorepr \lepr}
    {\hconf \lstore \le \lredexstep \hconf \lstorepr \lepr}
    \quad
    \inferrule
    {\ljudg \lstorety \lGamma {\lstorepr[2]}
      {\lectxt \hw {\var \square \lty}} \ltypr \quad
     \hconf \lstore \le
     \lredexstep
     \hconf \lstorepr \lepr}
    {\hconf {\lstorejoin {\lstorepr[2]} \lstore} {\lectxt \hw \le}
     \lredexstep
     \hconf {\lstorejoin {\lstorepr[2]} \lstorepr} {\lectxt \hw \lepr}}

    \inferrule
    {\ljudg \lstorety \lGamma \lstore \le \lty \quad
     \hconf \lstore \le \lredexstep \hconf \lstorepr \lepr \quad
     \ljudg \lstoretypr \lGamma \lstorepr \lepr \lty}
    {\hconf \lemptystore {\lsharee \lstore \lstorety \le}
     \lredexstep
     \hconf \lemptystore {\lsharee \lstorepr \lstoretypr \lepr}}
 \end{smathpar}
 \caption{Linear Language: Operational Semantics}
 \label{fig:l:opsem}
\end{figure}

\shortref{fig:l:opsem} gives a small-step operational semantics for
the internal terms of $\llang$. We separate the head reductions
$(\lheadredexstep)$ from reductions in depth $(\lredexstep)$. The head
reduction of the linear types of the core language do not involve the
store and are standard. For the store primitives of
\shortref{fig:l:store-surface-static-sem} acting on $\lemptyboxty,\lboxedty \lty$,
we reuse the isomorphism notation to emphasize that the related
primitives are inverses of each other.

There are several reduction rules for
$\lsrccopye {(\lsrcsharee \le)}$, one for each type connective. These
reductions perform a deep copy of the value, stopping only on ground
data ($\lunite$), function values, and shared sub-terms: when copying
a $\lbangty {\lbangty \lty}$ into a $\lbangty \lty$, there is no need
for a deep copy. When it encounters a location,
$\lsrccopye {(\lsrcsharee \lloc)}$ reduces to a new allocation.  If
the location contains a value, the new location is filled with a copy of
this value.

The copying rule for functions performs a copy of the local store
$\lstore$ of the shared function. The locations in $\lstore$ are bound
on the left-hand-side of the reduction, and free on the
right-hand-side: this reduction step allocates fresh locations, and
the store typing of the term changes from $\lemptystorety$ on the left
to $\lstorety$ on the right. The fact that reduction changes the store
typing is not unique to this rule, it is also the case when directly
copying locations. In ML languages with references, the store only
grows during reduction.  That is not the case for our linear store: our
reduction may either allocate new locations or free existing ones.

We define a grammar of (deterministic) reduction contexts, which contain
exactly one hole $\square$ in evaluation position. However, we only
define \emph{linear} contexts $\lectxt$ that do not share their hole:
we need a specific treatment of the $\lsharee \lstore \lstorety \le$
reduction. Its subterm $\le$ is reduced in the local store $\lstore$,
but may create or free locations in the store; so we need to update
the local store and its store typing during the reduction.

\begin{theorem}[Progress]
  \label{thm:l:progress}
  If $\ljudg \lstorety \lGamma \lstore \le \lty$, then either $\le$ is
  a value $\lv$ or there exists $\hconf \lstorepr \lepr$ such that
  $\hconf \lstore \le \lredexstep \hconf \lstorepr \lepr$.
\end{theorem}

\begin{proof}
  By induction on the typing derivation of $\le$, using induction
  hypothesis in the evaluation order corresponding to the structure of
  contexts $\lectxt$. If one induction hypothesis returns a reduction,
  we build a bigger reduction $(\lredexstep)$ for the whole term. If
  all induction hypotheses return a value, the proof depends on
  whether the head term-former is an introduction/construction form or
  an elimination/destruction form. An introduction form whose subterms
  are values is a value. For elimination forms, we use
  \fullref{lem:l:inversion} on the eliminated subterm (a value), to
  learn that it starts with an introduction form, and thus forms
  a head redex with the head elimination form, so we build a head
  reduction $(\lheadredexstep)$.
\end{proof}

\begin{lemma}[Non-store-escaping substitution principle]
  \label{lem:l:substitution}
  If
  \begin{smathpar}

    \ljudg \lstoretypr {\lGamma, \var \lx \lty} \lstorepr \le \ltypr

    \ljudg \lstorety \lGammapr \lstore \lv \lty

    \ctxjoin \lGamma \lGammapr

    \lx \notin \lstorety
  \end{smathpar}
  then
  \begin{smathpar}
    \ljudg
      {\storectxjoin \lstorety \lstoretypr}
      {\ctxjoin \lGamma \lGammapr}
      {\lstorejoin \lstore \lstorepr}
      {\subst \le \lv \lx}
      \ltypr
  \end{smathpar}
\end{lemma}

\begin{proof}
  The proof, summarized below, proceeds by induction on the typing
  derivation of $\le$.

  Most cases need an additional case analysis on whether the
  substituted type $\lty$ is a duplicable type of the form
  $\lbangty \ltypr[2]$, as it influence whether it may appear in zero
  or several subterms of $\le$. (This is a price to pay for
  contraction and weakening happening in all rules for convenience,
  instead of being isolated in separate structural rules.)

  For example, in the variable case, $\le$ may be the variable $\lx$
  itself, in which case we know that $\lGamma$ is empty and conclude
  immediately. But $\le$ may also be another variable $\ly$ if $\lx$
  is duplicable and has been dropped. In that case, we perform an
  inversion (\shortref{lem:l:inversion}) on the $\lv$ premise to learn
  that $\lstorety$ is empty and $\lGammapr$ is duplicable, and can thus
  use \fullref{lem:l:weakening}.

  In the $\lpaire \leone \letwo$ case, if $\lx$ is a linear variable
  it only occurs in one subterm on which we apply our induction
  hypothesis. If $\lx$ is duplicable, inversion on the $\lv$ premises
  again tells us that $\lGammapr$ is duplicable. We know by assumption
  that $\ctxjoin {(\ctxjoin \lGammaone \lGammatwo)} \lGammapr$; because
  $\lGammapr$ is duplicable, \Not\VeryShort{we can deduce from
  \fullref{lem:l:join-assoc-commut} that}{} the
  $\ctxjoin \lGammai \lGammapr$ are also defined, which let us apply an
  induction hypothesis on both subterms $\lei$. To conclude, we need
  the computation
  \begin{smathpar}
    \begin{array}{l@{~}l@{~}l}
      &   & \ctxjoin {(\ctxjoin \lGammaone \lGammapr)} {(\ctxjoin \lGammatwo \lGammapr)}
      \\
      & = & \ctxjoin {\ctxjoin \lGammaone \lGammatwo} {(\ctxjoin \lGammapr \lGammapr)}
      \\
      & = & \ctxjoin {\ctxjoin \lGammaone \lGammatwo} \lGammapr
    \end{array}
  \end{smathpar}
  which again comes from duplicability of $\lGammapr$.

  The assumption $\lx \notin \lstorety$ enforces that the resource
  $\lx$ is consumed in the term $\le$ itself, not in one of the values
  $\stextconf {} \lloc \lstore\lv$ in the store: otherwise $\lx$ would
  appear in the store typing $\lalive \lloc \lGamma \lstorety$ of this
  location in $\lstorety$. It is used in the case where $\le : \lty$
  is a full location $\lloc : \lboxedty \ltypr$. If $\lx$ could appear
  in the value of $\lloc$ in the store, we would have substitute it in
  the store as well -- in our substitution statement, only the term is
  modified. Here we know that this value is unused, so it has
  a duplicable type and we can perform an inversion in the other
  cases.
\end{proof}

\begin{lemma}[Context decomposition]
  \label{lem:l:decomposition}
  If
  $\ljudg \lstoretypr \lGammapr \lstorepr {\lectxt \hw {\var \square \lty}} \ltypr$
  holds, then
  $\ljudg {\lstoretypr[2]} {\lGammapr[2]} {\lstorepr[2]} {\lectxt \hw \le} \ltypr$
  holds if and only if there exists $\lstorety, \lGamma, \lstore$ such that
  $\lstoretypr[2] = \storectxjoin \lstorety \lstoretypr$,
  $\lGammapr[2] = \ctxjoin \lGamma \lGammapr$,
  $\lstorepr[2] = \lstorejoin \lstore \lstorepr$
  and $\ljudg \lstorety \lGamma \lstore \le \lty$.
\end{lemma}

\begin{theorem}[Subject reduction for $\llang$]
  \label{thm:l:subject-reduction}
  If $\ljudg \lstorety \lGamma \lstore \le \lty$
  and $\hconf \lstore \le \lredexstep \hconf \lstorepr \lepr$,
  then there exists a (unique) $\lstoretypr$ such that
  $\ljudg \lstoretypr \lGamma \lstorepr \lepr \lty$.
\end{theorem}

\begin{proof}
  The proof is done by induction on the reduction derivation.

  The head-reduction rules involving substitutions rely on
  \fullref{lem:l:substitution}; note that in each of them, for example
  $\lappe {(\lfune \lx \ltypr \lepr)} {\lepr[2]}$, the substituted
  variable $\lx$ is bound in the term $\le$, and thus does not appear
  in the store $\lstore$: the non-store-escaping hypothesis holds.

  For the copy rule and the store operators, we build a valid
  derivation for the reduced configuration by inverting the typing
  derivation of the reducible configuration.

  In the non-head-reduction cases, the $\lsrcsharee {}$ case is by
  direction, and the context case $\lectxt \hw \le$ uses
  \fullref{lem:l:decomposition} to obtain a typing derivation for
  $\le$, and the same lemma again rebuild a derivation of the reduced
  term $\lectxt \hw \lepr$.
\end{proof}

\endgroup}

\Not\VeryShort{\SectionInternalLLReduction}{\begingroup
In
  \Appendices
    {\fullref{appendix:internal-linear-reduction}}
    {the long version of this work}
we give a reduction relation between linear configurations
$\hconf \lstore \le \lredexstep \hconf \lstorepr \lepr$
and prove a subject reduction result.

\begin{theorem}[Subject reduction for $\llang$]
  If $\ljudg \lstorety \lGamma \lstore \le \lty$
  and $\hconf \lstore \le \lredexstep \hconf \lstorepr \lepr$,
  then there exists a (unique) $\lstoretypr$ such that
  $\ljudg \lstoretypr \lGamma \lstorepr \lepr \lty$.
\end{theorem}
\endgroup}

\section{Multi-language semantics}
\label{sec:multi-language}

To formally define our multi-language semantics we create a combined
language $\ullang$ which lets us compose term fragments from both
$\ulang$ and $\llang$ together, and we give an operational semantics
to this combined language. Interoperability is enabled by specifying
how to transport values across the language boundaries.

Multi-language systems in the wild are not defined in this
way: both languages are given a semantics, by interpretation or
compilation, in terms of a shared lower-level language (C, assembly,
the JVM or CLR bytecode, or Racket's core forms), and the two
languages are combined at that level. Our formal multi-language
description can be seen as a model of such combinations, that gives
a specification of the expected observable behavior of this language
combination.

Another difference from multi-languages in the wild is our use of very
fine-grained language boundaries: a term written in one language can
have its subterms written in the other, provided the type-checking
rules allow it. Most multi-language systems, typically using Foreign
Function Interfaces, offer coarser-grained composition at the level of
compilation units. Fine-grained composition of existing languages, as
done in the Eco project~\citep*{eco}, is difficult because of semantic
mismatches. In \Appendices{\fullref{sec:examples}}{the full version of this
  work} we demonstrate that fine-grained
composition is a rewarding language design, enabling new programming
patterns.

\subsection{Lump Type and Language Boundaries}
\label{subsec:ul:lump-and-boundaries}

\begin{figure}
  \begin{smathpar}
    \begin{array}{l@{~}c@{~}r@{~}ll}
      \mbox{\textit{Types}}
      & \uty \bnfalt \lty & &
      \\[1pt]
      & \uty & & \text{(unchanged from \shortref{fig:u:syntax})}
      \\
      & \lty & \bnfadd & \dots \bnfalt \llumpty\uty
      \\[3pt]
    \end{array}

    \begin{array}{l@{~}c@{~}r@{~}ll}
      \mbox{\textit{Values}}
      & \uv \bnfalt \lv & &
      \\[2pt]
      & \uv & &  \text{(unchanged from \shortref{fig:u:syntax})}
      \\
      & \lv & \bnfadd & \dots \bnfalt \llumpe \uv
      \\[3pt]
    \end{array}

    \begin{array}{l@{~}c@{~}r@{~}ll}
      \mbox{\textit{Expressions}}
      & \ue \bnfalt \le & &
      \\[2pt]
      & \ue & \bnfadd & \dots \bnfalt \ULe \lstore \lstorety \le
      \\
      & & & \text{with}\quad
            \ULenoconf \le
            \;\defeq\;
            \ULe \lemptystore \lemptystorety \le
      \\
      & \le & \bnfadd & \dots \bnfalt \LUe \ue
      \\[3pt]
    \end{array}

    \begin{array}{l@{~}c@{~}r@{~}ll}
      \mbox{\textit{Contexts}}
      & \ulGamma
      & \bnfdef
      & \ulemptyGamma
        \bnfalt
        \ulGamma, \var \ux \uty
        \bnfalt
        \ulGamma, \ualpha
        \bnfalt
        \ulGamma, \var \lx \lty
    \end{array}
  \end{smathpar}

  \vspace{3pt} Typing rules
  \quad
  \fbox{\small$\ulujudg \ulGamma \ue \uty$}
  \quad
  \fbox{\small$\ulljudg \lstorety \ulGamma \lstore \le \lty$}

  with
  \quad
  $\bareuljudg \ulGamma \le \lty
    \;\defeq\; \ulljudg \lemptystorety \ulGamma \lemptystore \le \lty$

  \begin{smathpar}
  \begin{array}{l}
  \text{(Typing rules of $\ujudg \uGamma \ue \uty$
        reused, with mixed context $\ulbGamma$)}
  \\
  \text{(Typing rules of $\ljudg \lstorety \lGamma \lstore \le \lty$
        reused, with mixed context $\ulGamma$)}
  \end{array}
  \\
  \infer
  {\ulujudg \ulbGamma \ue \uty}
  {\ulljudg \lemptystorety \ulbGamma \lemptystore {\LUe \ue}
    {\llumpedty \uty}}

  \infer
  {\ulljudg \lstorety \ulbGamma \lstore \le {\llumpedty \uty}}
  {\ulujudg {\ulbGamma} {\ULe \lstore \lstorety \le} \uty}
  \end{smathpar}
  
  \vspace{3pt} Reduction rules
  \begin{smathpar}
  \text{(Reduction rules of $\ulang$ and $\llang$ reused unchanged)}

  \infer
  {\ue \uredexstep \uepr}
  {\LUe \ue \lredexstep \LUe \uepr}

  \infer
  { }
  {\LUe \uv \lheadredexstep \llumpe \uv}

  \infer
  { }
  {\ULe \lemptystore \lemptystorety {\lsrcsharee {\llumpe \uv}} \uredexstep \uv}

  \infer
  {\ulljudg \lstorety \ulGamma \lstore \le \lty
   \+
   \conf \lstore \le \lredexstep \conf \lstorepr \lepr
   \+
   \ulljudg \lstoretypr \ulGamma \lstorepr \lepr \lty}
  {\ULe \lstore \lstorety \le \uredexstep \ULe \lstorepr \lstoretypr \lepr}
  \end{smathpar}
  \caption{Multi-language: Lump and Boundaries}
  \label{fig:ul:static-sem}
\end{figure}
\begin{version}{\VeryShort}
\begin{figure}
  static compatibility
  \quad
  \fbox{\small{$\ulSigma \ulwf \uty \ulcompat \lty$}}
  \begin{smathpar}
    \infer
    { }
    {\ulSigma \ulwf \uty \ulcompat {\llumpedty \uty}}

    \infer
    {\ulSigma \ulwf \utyone \ulcompatbang \ltyone\\
     \ulSigma \ulwf \utytwo \ulcompatbang \ltytwo}
    {\ulSigma \ulwf
      \usumty \utyone \utytwo
      \ulcompatbangpar{\losumty \ltyone \ltytwo}}

    \infer
    {\ulSigma \ulwf \uty \ulcompatbang \lty\\
     \ulSigma \ulwf \utypr \ulcompatbang \ltypr}
    {\ulSigma \ulwf
      \ufunty \uty \utypr
      \ulcompatbangpar{\lofunty {\lbangty \lty} {\lbangty \ltypr}}}

    \infer
    {\ulSigma \ulwf \uty \ulcompatbang \lty}
    {\ulSigma \ulwf \uty \ulcompatbang {\lbangty \lty}}

    \infer
    {\ulSigma \ulwf \uty \ulcompatbang \lty}
    {\ulSigma \ulwf \uty \ulcompatbangpar {\lboxedty \lty}}
  \end{smathpar}

  value conversion
  \quad
  \fbox{\small{$\uv \ulcompate \lty \lv$}}
  \begin{smathpar}
    %
    \infer
    { }
    {\uv \ulcompatsharee {\llumpty \uty} {\llumpe \uv}}

    \infer
    {\uv \ulcompate {\lbangty \ltyi} {\lsharee \lstore \lstorety \lv}}
    {\usume i \uv
     \ulcompate
       {\lbangty{(\losumty \ltyone \ltytwo)}}
       {\lsharee \lstore \lstorety {\lsume i \lv}}}

    \infer
    {\uty \ulcompat \lty\\
     \utypr \ulcompat \ltypr}
    {\ue
     \ulcompatrighte
       {\lbangty {(\lofunty {\lbangty\lty} {\lbangty\ltypr})}}
       {\lsrcsharee
         {\lfune \lx {\lbangty \lty}
           {\LUmpe \ltypr {\uappe \ue {\UmpLenoconf \lty \lx}}}}}}
    \qquad
    \infer
    {\uty \ulcompat \lty\\
     \utypr \ulcompat \ltypr}
    {\ufune \ux \uty
      {\UmpLenoconf \ltypr
        {\lappe {\lsrccopye \le}
          {\LUmpe \lty \ux})}}
     \ulcompatlefte
       {\lbangty {(\lofunty \lty \ltypr)}}
       {\le}}

    \infer
    {\uv \ulcompatsharee \lty \lv}
    {\uv \ulcompatsharee {\lbangty \lty} {(\lsrcsharee \lv)}}

    \infer
    {\uv \ulcompate {\lbangty \lty} {\lsharee \lstore \lstorety \lv}}
    {\uv \ulcompate
      {\lbangty {\lboxedty \lty}}
      {\lsharee
        {\stsingletonconf \lloc \lstore \lv}
        {\lalive \lloc \lemptyGamma \lstorety \lty}
        \lloc}}

    \infer
    {\uv
     \ulcompate{\lbangty {\subst \lty {\lmuty \lalpha \lty} \lalpha}}
     \lsharee \lstore \lstorety \lv}
    {\ufolde {\umuty \ualpha \uty} \uv
     \ulcompate{\lbangty {\lmuty \lalpha \lty}}
     {\lsharee \lstore \lstorety {(\lfolde {\lmuty \lalpha \lty} \lv})}}
  \end{smathpar}

  \caption{Interoperability: static and dynamic semantics (excerpt)}
  \label{fig:l:interoperability-semantics-excerpt}
\end{figure}
\end{version}

The core components the multi-language semantics are shown
\shortref{fig:ul:static-sem}---the communication of values from one
language to the other will be described in the next section. The
multi-language $\ullang$ has two distinct syntactic categories of
types, values, and expressions: those that come from $\ulang$ and
those that come from $\llang$. Contexts, on the other hand, are mixed,
and can have variables of both sorts.  For a mixed context $\ulGamma$,
the notation $\lbangty \ulGamma$ only applies $(\lbangty)$ to its
linear variables.

The typing rules of $\ulang$ and $\llang$ are imported into our
multi-language system, working on those two separate categories of
program. They need to be extended to handle mixed contexts $\ulGamma$
instead of their original contexts $\uGamma$ and $\lGamma$. In the
linear case, the rules look exactly the same. In the ML case, the
typing rules implicitly duplicate all the variables in the context.
It would be unsound to extend them to arbitrary linear variables, so
they use a duplicable context $\lbangty \ulGamma$.

To build interesting multi-language programs, we need a way to insert
a fragment coming from a language into a term written in another. This
is done using \emph{language boundaries}, two new term formers $\LUe \ue$
and $\ULe \lstore \lstorety \le$ that inject an ML term into the
syntactic category of linear terms, and a linear configuration into
the syntactic category of ML terms.

Of course, we need new typing rules for these term-level
constructions, clarifying when it is valid to send a value from
$\ulang$ into $\llang$ and vice versa.  It would be incorrect to allow
sending any type from one language into the other---for instance, by
adding the counterpart of our language boundaries in the syntax of
types---since values of linear types must be uniquely owned so they
cannot possibly be sent to the ML side as the ML type system cannot
enforce unique ownership.

On the other hand, any ML value could safely be sent to the linear
world. For closed types, we could provide a corresponding linear type
($\uunitty$ maps to $\lbangty \lunitty$, etc.), but an ML value may
also be typed by an abstract type variable $\ualpha$, in which case we
can't know what the linear counterpart should be. Instead of trying to
provide translations, we will send any ML type $\uty$ to the
\emph{lump type} $\llumpty \uty$, which embeds ML types into linear
types. A lump is a blackbox, not a type translation: the linear
language does not assume anything about the behavior of its
values---the values of $\llumpty \uty$ are of the form $\llumpe \uv$,
where $\uv : \uty$ is an ML value that the linear world cannot
use. More precisely, we only propagate the information that ML values
are all duplicable by sending $\uty$ to $\llumpedty \uty$.

The typing rules for language boundaries insert lumps when going
from $\ulang$ to $\llang$, and remove them when going back from
$\llang$ to $\ulang$. In particular, arbitrary linear types cannot
occur at the boundary, they must be of the form $\llumpedty \uty$.

Finally, boundaries have reduction rules: a term or configuration
inside a boundary in reduction position is reduced until it becomes
a value, and then a lump is added or removed depending on the boundary
direction. Note that because the $\lv$ in $\ULe \lstore \lstorety \lv$
is at a duplicable type $\llumpedty \uty$, we know by inversion that
the store is empty.

\subsection{Interoperability: Static Semantics}
\label{subsec:ul:interop}

If the linear language could not interact with lumped values at all,
our multi-language programs would be rather boring, as the only way
for the linear extension to provide a value back to ML would be to
have received it from $\ulang$ and pass it back unchanged (as in the
lump embedding of~\citet*{matthews2009operational}). To provide
a real interaction, we provide a way to extract values out of a lump
$\llumpedty \uty$, use it at some linear type $\lty$, and put it back
in before sending the result to $\ulang$.

\newcommand{\SectionInteropStatic}{\begingroup
\begin{figure}
  \begin{smathpar}
    \begin{array}{l@{~}c@{~}r@{~}ll}
      \mbox{\textit{Interoperability context}}
      & \ulSigma & \bnfdef & \ulempty \bnfalt \ulSigma, \ualpha \ulcompatbang \lbeta
    \end{array}
  \end{smathpar}

  \vspace{3pt}
  Compatibility relation
  \quad
  \fbox{\small$\ulSigma \ulwf \uty \ulcompat \lty$}
  \quad
  with $\quad\uty \ulcompat \lty \;\defeq\; \ulempty \ulwf \uty \ulcompat \lty$
  \vspace{3pt}

  \begin{smathpar}
    \infer
    { }
    {\ulSigma \ulwf \uunitty \ulcompatbang \lunitty}

    \infer
    {\ulSigma \ulwf \utyone \ulcompatbang \ltyone\\
     \ulSigma \ulwf \utytwo \ulcompatbang \ltytwo}
    {\ulSigma \ulwf
      \upairty \utyone \utytwo
      \ulcompatbangpar{\lopairty \ltyone \ltytwo}}

    \infer
    {\ulSigma \ulwf \utyone \ulcompatbang \ltyone\\
     \ulSigma \ulwf \utytwo \ulcompatbang \ltytwo}
    {\ulSigma \ulwf
      \usumty \utyone \utytwo
      \ulcompatbangpar{\losumty \ltyone \ltytwo}}

    \infer
    {\ulSigma \ulwf \uty \ulcompatbang \lty\\
     \ulSigma \ulwf \utypr \ulcompatbang \ltypr}
    {\ulSigma \ulwf
      \ufunty \uty \utypr
      \ulcompatbangpar{\lofunty {\lbangty \lty} {\lbangty \ltypr}}}

    \infer
    { }
    {\ulSigma \ulwf \uty \ulcompat {\llumpedty \uty}}

    \infer
    {\ulSigma \ulwf \uty \ulcompatbang \lty}
    {\ulSigma \ulwf \uty \ulcompatbang {\lbangty \lty}}

    \infer
    {\ulSigma \ulwf \uty \ulcompatbang \lty}
    {\ulSigma \ulwf \uty \ulcompatbangpar {\lboxedty \lty}}

    \infer
    {\ulSigma, \ualpha \ulcompatbang \lbeta \ulwf
       \uty \ulcompatbang \lty}
    {\ulSigma \ulwf
      \umuty \ualpha \uty
      \ulcompatbangpar{\lmuty \lbeta \lty}}

    \infer
    {(\ualpha \ulcompatbang \lbeta) \in \ulSigma}
    {\ulSigma \ulwf \ualpha \ulcompatbang \lbeta}
  \end{smathpar}

  \vspace{3pt} Interoperability primitives and derived constructs:
  \begin{smathpar}
  \bidirliniso
    {\llumpedty \uty}
    {\lty}
    {\unlump \lty {}}
    {\lump \lty {}}
  \quad
  \text{whenever}
  \quad
  \ulempty \ulwf \uty \ulcompat \lty

  \LUmpe \lty \ue \;\defeq\; \unlump \lty {\LUe \ue}

  \UmpLenoconf \lty \le \;\defeq\; \ULenoconf {\lump \lty \le}
  \end{smathpar}
  \caption{Multi-language: Static Interoperability Semantics}
  \label{fig:ul:static-interop-sem}
\end{figure}

The correspondence between intuitionistic types $\uty$ and linear
types $\lty$ is specified by a heterogeneous \emph{compatibility
  relation} $\uty \ulcompat \lty$ defined in
\fullref{fig:ul:static-interop-sem}. The specification of this
relation is that if $\uty \ulcompat \lty$ holds, then the space of
values of $\llumpedty \uty$ and $\lty$ are isomorphic: we can convert
back and forth between them. When this relation holds, the
term-formers $\lump \lty {}$ and $\unlump \lty {}$ perform the
conversion. (The position of the index $\lty$ emphasizes that the
\emph{input} $\le$ of $\lump \lty \le$ has type $\lty$, while the
\emph{output} of $\unlump \lty \le$ has type $\lty$.)

For example, we have
$\llumpedty {(\ufunty \uty \utypr)} \ulcompat \lbangty {(\lofunty {\llumpedty \uty} {\llumpedty \utypr})}$. Given
a lumped ML function, we can unlump it to see it as a linear
function. We can call it from the linear side, but have to pass it
a duplicable argument since an ML function may duplicate its
argument. Conversely, we can convert a linear function into a lumped
function type to pass it to the ML side, but it has to have
a duplicable return type since the ML side may freely share the return
value. 

Our $\lump \lty {}$ and $\unlump \lty {}$ primitives are only indexed
by the linear type $\lty$, because a compatible ML type $\uty$ can be
uniquely recovered, as per the following result.

\begin{lemma}[Determinism of the compatibility relation]
  \label{lem:ul:determinism}
  If $\uty \ulcompat \lty$ and $\utypr \ulcompat \lty$ then
  $\uty = \utypr$.
\end{lemma}

\begin{proof}
  By induction on the syntax of $\lty$: the judgment
  $\ulSigma \ulwf \uty \ulcompat \lty$ is syntax-directed in its
  $\lty$ component.
\end{proof}

Note that the converse property does not hold: for a given $\uty$,
there are many $\lty$ such that $\uty \ulcompat \lty$. For example, we
have $\uunitty \ulcompat \lbangty \lunitty$ but also
$\uunitty \ulcompat \lbangty {\lbangty \lunitty}$. This corresponds to
the fact that the linear types are more fine-grained, and make
distinctions (inner duplicability, dereference of full locations) that
are erased in the ML world. The $\uty \ulcompat \llumpedty \uty$ case
also allows you to (un)lump as deeply or as shallowly as you need:
$\upairty \utyone {(\upairty \utytwo \utythree)}$ is compatible with both
$\lbangty
  {(\lopairty
    {\llumpedty \utyone}
    {\llumpedty {\upairty \utytwo \utythree}})}
$
and
$\lbangty
  {(\lopairty
    {\llumpedty \utyone}
    {(\lopairty {\llumpedty \utytwo} {\llumpedty \utythree})})}
$.
We could not systematically translate the complete type $\uty$, as
type variables cannot be translated and need to remain lumped.  Allowing
lumps to ``stop'' the translation at arbitrary depth is a natural
generalization.

\begin{lemma}[Substitution of recursive hypotheses]
  \label{lem:ul:interop-hypot-subst}
  If
  $
    \ulSigma, \ualpha \ulcompat \lbangty \lbeta \ulwf \uty \ulcompat \lty 
  $,
  $
    \ualpha \notin \lty
  $, and
  $
    \ulSigma \ulwf \utypr \ulcompat \lbangty\ltypr
  $
  then
  $
    \ulSigma \ulwf
    \subst \uty \utypr \ualpha
    \ulcompat \subst \lty \ltypr \lbeta
  $.
\end{lemma}

\begin{proof}
  By induction on the $\uty \ulcompat \lty$ derivation. There are two
  leaf cases: the case recursive hypotheses, which is immediate, and
  the case of lump $\uty \ulcompat \llumpedty \uty$. In this latter
  case, notice that $\uty$ is a type of $\ulang$, so in particular it
  does not contain the variable $\lbeta$; and we assumed
  $\ualpha \notin \lty$ so we also have $\ualpha \notin \uty$, so
  $\subst \uty \utypr \ualpha
  = \uty
  \ulcompat \llumpedty \uty
  = \subst \lty \ltypr \lbeta
  $.
\end{proof}
\endgroup}

\Not\VeryShort{\SectionInteropStatic}{\begingroup
The correspondence between intuitionistic types $\uty$ and linear
types $\lty$ is specified by a heterogeneous \emph{compatibility
  relation} $\uty \ulcompat \lty$ -- defined in full in
  \Appendices
    {\fullref{appendix:interop-static}}
    {\fullref{fig:l:interoperability-semantics-excerpt}}
. The specification of this
relation is that if $\uty \ulcompat \lty$ holds, then the space of
values of $\llumpedty \uty$ and $\lty$ are isomorphic: we can convert
back and forth between them. When this relation holds, the
term-formers $\lump \lty {}$ and $\unlump \lty {}$ perform the
conversion.
\endgroup}

The term $\LUe \ue$ turns a $\ue : \uty$ into a lumped type
$\llumpedty \uty$, and we need to unlump it with some
$\unlump \lty {}$ for a compatible $\uty \ulcompat \lty$ to interact
with it on the linear side. It is common to combine both operations
and we provide syntactic sugar for it: $\LUmpe \lty \ue$. Similarly
$\UmpLenoconf \lty \le$ first lumps a linear term then sends the result
to the ML world.

\subsection{Interoperability: Dynamic Semantics}

\newcommand{\SectionInteropDynamic}{\begingroup
\begin{figure}
  \begin{smathpar}
    %
    \uunite \ulcompatsharee \lunitty \lunite

    \infer
    {\begin{array}{l}
       \uvone \ulcompate {\lbangty \ltyone} {\lsharee \lstoreone \lstoretyone \lvone}
       \\
       \uvtwo \ulcompate {\lbangty \ltytwo} {\lsharee \lstoretwo \lstoretytwo \lvtwo}
      \end{array}
      \\
     \begin{array}{l}
       \locs{\lstoreone} = \locs{\lstoretyone} = \locs{\lvone} \\
       \locs{\lstoretwo} = \locs{\lstoretytwo} = \locs{\lvtwo}
     \end{array}}
    {\upaire \uvone \uvtwo
     \ulcompate
       {\lbangty{(\lopairty \ltyone \ltytwo)}}
       {\lsharee
         {\lstorejoin \lstoreone \lstoretwo}
         {\storectxjoin \lstoretyone \lstoretytwo}
         {\lpaire \lvone \lvtwo}}}

    \infer
    {\uv \ulcompate {\lbangty \ltyi} {\lsharee \lstore \lstorety \lv}}
    {\usume i \uv
     \ulcompate
       {\lbangty{(\losumty \ltyone \ltytwo)}}
       {\lsharee \lstore \lstorety {\lsume i \lv}}}

    \infer
    {\uty \ulcompat \lty\\
     \utypr \ulcompat \ltypr}
    {\ue
     \ulcompatrighte
       {\lbangty {(\lofunty {\lbangty\lty} {\lbangty\ltypr})}}
       {\lsrcsharee
         {\lfune \lx {\lbangty \lty}
           {\LUmpe \ltypr {\uappe \ue {\UmpLenoconf \lty \lx}}}}}}

    \infer
    {\uty \ulcompat \lty\\
     \utypr \ulcompat \ltypr}
    {\ufune \ux \uty
      {\UmpLenoconf \ltypr
        {\lappe {\lsrccopye \le}
          {\LUmpe \lty \ux})}}
     \ulcompatlefte
       {\lbangty {(\lofunty \lty \ltypr)}}
       {\le}}

    \infer
    { }
    {\uv \ulcompatsharee {\llumpty \uty} {\llumpe \uv}}

    \infer
    {\uv \ulcompatsharee \lty \lv}
    {\uv \ulcompatsharee {\lbangty \lty} {(\lsrcsharee \lv)}}

    \infer
    {\uv \ulcompate {\lbangty \lty} {\lsharee \lstore \lstorety \lv}}
    {\uv \ulcompate
      {\lbangty {\lboxedty \lty}}
      {\lsharee
        {\stsingletonconf \lloc \lstore \lv}
        {\lalive \lloc \lemptyGamma \lstorety \lty}
        \lloc}}

    \infer
    {\uv
     \ulcompate{\lbangty {\subst \lty {\lmuty \lalpha \lty} \lalpha}}
     \lsharee \lstore \lstorety \lv}
    {\ufolde {\umuty \ualpha \uty} \uv
     \ulcompate{\lbangty {\lmuty \lalpha \lty}}
     {\lsharee \lstore \lstorety {(\lfolde {\lmuty \lalpha \lty} \lv})}}

    \bidirlheadredexstep
      {\hconf \lemptystore {\lsrcsharee {\llumpe \uv}}}
      {\hconf \lemptystore \lv}
      {\unlump \lty {}}
      {\lump \lty {}}
    \begin{array}{ll}
      \text{whenever} & \uv \ulcompatrighte \lty \lv
      \\
      \\
      \\
      \text{whenever} & \uv \ulcompatlefte \lty \lv
    \end{array}
  \end{smathpar}

  \caption{Multi-language: Dynamic Interoperability Semantics}
  \label{fig:ul:dynamic-interop-sem}
\end{figure}

We were careful to define the compatibility relation such that
$\uty \ulcompat \lty$ only holds when $\llumpedty \uty$ and $\lty$ are
isomorphic, in the sense that any value of one can be converted into
a value of another. \shortref{fig:ul:dynamic-interop-sem} defines
the operational semantics of the lumping and unlumping operations
precisely as realizing these isomorphisms. For concision, we specify
the isomorphisms as relations, following the inductive structure of
the compatibility judgment itself. We write $(\ulcompate{})$ when
a rule can be read bidirectionally to convert in either directions
(assuming the same direction holds of the premises), and
$(\ulcompatlefte{})$ or $(\ulcompatrighte{})$ for rules that only
describe how to convert values in one direction.

\begin{lemma}[Substitution of polymorphic variables]
  \label{lem:ul:interop-poly-subst}
  If $\ulSigma \ulwf \uty \ulcompat \lty$ and $\ualpha \notin \ulSigma$,
  then
  $\ulSigma \ulwf
   \subst \uty \utypr \ualpha
   \ulcompat
   \subst \lty \utypr \ualpha
  $
  and their lumping operations coincide on all values.
\end{lemma}

\begin{proof}
  By induction on $\uty \ulcompat \lty$. In the variable leaf case, we
  know $\ualpha \notin \ulSigma$. In the lump leaf case
  $\uty \ulcompat \llumpedty \uty$, the goal
  $\subst \uty \utypr \ualpha \ulcompat \llumpedty {\subst \uty \utypr \ualpha}$
  is immediate.
\end{proof}

\begin{theorem}[Value translations are functional]
  \label{thm:ul:compat-value-sem}
  If $\uty \ulcompat \lty$, then for any closed value $\uv : \uty$
  there is a unique $\lv : \lty$ such that $\uv \ulcompatrighte \lty \lv$,
  and conversely for any closed value $\lv : \lty$ there is a unique
  $\uv : \uty$ such that $\uv \ulcompatlefte \lty \lv$.
\end{theorem}

\begin{proof}
  Remark that in the statement of the term, when we quantify over all
  closed values $\lv : \lty$, we implicitly assume that in the general
  case values of $\lv$ live in the empty global store -- otherwise we
  would have a value of the form
  $\ljudg \lstorety \lemptyGamma \lstore \lv \lty$. This is valid
  because all types $\lty$ in the image of the type-compatibility
  relation are duplicable types of the form $\lbangty \ltypr$, so by
  \fullref{lem:l:inversion} we know that $\lv$ is in fact of the form
  $\lsharee \lstore \lstorety \lepr$, living in the empty store.)

  The two sides of the result are proved simultaneously by induction
  on $\uty \ulcompat \lty$, using inversion to reason on the shapes of
  $\uv$ and $\lv$. Note that the inductive cases remain on closed
  values: the only variable-binder constructions,
  $\lambda$-abstractions, do not use the recursion hypothesis.

  In the recursive case
  $\umuty \ualpha \uty \ulcompat \lbangty {(\lmuty \lalpha \lty)}$, to
  use the induction hypothesis on the folded values we need to know
  that the unfolded types $\subst \uty {\umuty \ualpha \uty} \ualpha$
  and $\subst \lty {\lmuty \lalpha \lty} \lalpha$ are compatible. This
  is exactly \fullref{lem:ul:interop-hypot-subst}, using
  the hypothesis
  $\umuty \ualpha \uty \ulcompat \lbangty {(\lmuty \lalpha \lty)}$
  itself.
\end{proof}

\begin{lemma}[Lumping cancellation]
  \label{lem:lumping-cancellation}
  The lump conversions $\lump \lty {}$ and $\unlump \lty {}$
  cancel each other modulo $\beta\eta$. In particular,
  \begin{smathpar}
    \LUmpe \lty {\UmpLe \lty  \lv} =_{\beta\eta} \lv

    \UmpLenoconf \lty {\LUmpe \lty \uv} =_{\beta\eta} \uv 
  \end{smathpar}
\end{lemma}

\begin{proof}
  By induction on $\uty$, and then by parallel induction on the
  derivations of type compatibility and value compatibility. The
  parallel cases are symmetric by definition, only the function case
  $\lbangty {(\lofunty {\lbangty \lty} {\lbangty \ltypr})}
  $
  needs to be checked. A simple computation, using the induction
  hypothesis on the smaller types $\lbangty \lty$ and
  $\lbangty \ltypr$, shows that composing the two function
  translations gives an $\eta$-expansion -- plus the $\beta\eta$-steps
  from the induction hypotheses.

  Note that the $\beta$-$\eta$ rules used (for functions and
  $\lbangty \wild$) have been proved admissible for the logical
  relation, so in particular sound with respect with contextual
  equivalence.
\end{proof}

\paragraph{Implementation consideration} In a realistic implementation
of this multi-language system, we would expect the representation
choices made for $\ulang$ and $\llang$ to be such that, for some but
not all compatible pairs $\uty \ulcompat \lty$, the types $\uty$ and
$\lty$ actually have the exact same representation, making the
conversion an efficient no-op. An implementation could even restrict
the compatibility relation to accept only the pairs that can be
implemented in this way.  That is, it would reject some $\ullang$
programs, but the ``graceful interoperability'' result that is our
essential contribution would still hold.
\endgroup}

\Not\VeryShort{\SectionInteropDynamic}{\begingroup When the relation
  $\uty \ulcompat \lty$ holds, we can define a relation
  $\uv \ulcompate \lty \lv$ between the values of $\uty$ and the
  values of $\lty$ -- see 
  \Appendices
    {\fullref{appendix:interop-dynamic}}
    {the long version of this work}. It is
  functional in both direction: with our definition $\uv$ is uniquely
  determined from $\lv$ and conversely. We then define the reduction
  rule for (un)lumping: if $\uv \ulcompate \lty \lv$, then
  \vspace{-1em}
  \begin{smathpar}
    \hconf \lemptystore {\unlump \lty {(\lsrcsharee {\llumpe \uv}})}
    \lredexstep
    \hconf \lemptystore \lv
 
    \hconf \lemptystore {\lump \lty \lv}
    \lredexstep
    {\hconf \lemptystore {\lsrcsharee {\llumpe \uv}}} 
  \end{smathpar}\vspace{-2em}
\endgroup}

\subsection{Full Abstraction from \texorpdfstring{$\ulang$}{U}
  into \texorpdfstring{$\ullang$}{UL}}
\label{subsect:full-abstraction}

We can now state \VeryShort{}{and prove} the major meta-theoretical result of this
work, which is the proposed multi-language design extends the simple
language $\ulang$ in a way that provably has, in a certain sense, ``no
abstraction leaks''.

\begin{definition}[Contextual equivalence in $\ulang$]
  We say that $\ue, \uepr$ such that
  $\ujudg \uGamma {\ue, \uepr} \uty$ are \emph{contextually
    equivalent}, written $\ue \uctxeq \uepr$, if, for any expression
  context $\uctxt \hw \square$ such that
  $\ujudg \uemptyGamma {\uctxt \hw \ue} \uunitty$, the closed terms
  $\uctxt \hw \ue$ and $\uctxt \hw \uepr$ are equi-terminating.
\end{definition}

\begin{definition}[Contextual equivalence in $\ullang$]
  We say that $\ue, \uepr$ such that
  $\ulujudg \ulGamma {\ue, \uepr} \uty$ are \emph{contextually
    equivalent}, written $\ue \ulctxeq \uepr$, if, for any expression
  context $\ulctxt \hw \square$ such that
  $\ulujudg \ulemptyGamma {\ulctxt \hw \ue} \uunitty$, the closed terms
  $\ulctxt \hw \ue$ and $\ulctxt \hw \uepr$ are equi-terminating.
\end{definition}

\newcommand{\SectionFullAbstraction}{\begingroup
The proof of full abstraction is actually rather simple. It relies
on the idea, that we already mentioned in \fullref{subsec:l:store},
that linear state can be seen as either being imperatively mutated,
but also as a purely functional feature that just explicits memory
layout. In the absence of aliasing, we can give a purely functional
semantics to linear state operations---instead of the
store-modifying reduction semantics---and, in fact,
this semantics determines a translation from linear programs back
into pure ML programs. Those ML programs will not have the same
allocation behavior as the initial linear programs (in-place
programs won't be in-place anymore), but they are observably
equivalent in that they are equi-terminating and return the same
outputs from the same inputs.

\begin{figure}
  \begin{smathpar}
    \begin{array}{lcl}
      \funtrans \lunitty
      & \defeq
      & \uunitty
      \\
      \funtrans {\lopairty \ltyone \ltytwo}
      & \defeq
      & \upairty {\funtrans \ltyone} {\funtrans \ltytwo}
      \\[2pt]
      & \dots &
      \\[2pt]
      \funtrans {\lbangty \lty}
      & \defeq
      & \funtrans \lty
    \end{array}

    \begin{array}{lcl}
      \funtrans {\lemptyboxty}
      & \defeq
      & \uunitty
      \\
      \funtrans {\lboxedty \lty}
      & \defeq
      & \upairty \uunitty {\funtrans \lty}
      \\
      \funtrans {\llumpty \uty}
      & \defeq
      & \uty
      \\[2pt]
      \funtrans {\lmuty \lalpha \lty}
      & \defeq
      & \umuty {\ualpha_\lalpha} {\funtrans \lty}
      \\
      \funtrans \lalpha
      & \defeq
      & \ualpha_\lalpha
    \end{array}
  \end{smathpar}

\hrule

  \begin{smathpar}
    \begin{array}{l@{~}l@{~}l}
      \funtrans \lunite
      & \defeq
      & \uunite
      \\
      \funtrans {\lletunite \le \lepr}
      & \defeq
      & \uletunite {\funtrans \le} {\funtrans \lepr}
      \\
      \funtrans {\lfune \lx \lty \le}
      & \defeq
      & \ufune {\ux_\lx} {\funtrans \lty} {\funtrans \le}
      \\
      \funtrans {\lappe \le \lepr}
      & \defeq
      & \lappe {\funtrans \le} {\funtrans \lepr}
      \\
    \end{array}\qquad
    \begin{array}{l@{~}l@{~}l}
      \funtrans {\LUe \ue}
      & \defeq
      & \funtrans \ue
      \\[2pt]
      \funtrans {\lnewe \le}
      & \defeq
      & \uletunite {\funtrans \le} \uunite
      \\
      \funtrans {\lfreee \le}
      & \defeq
      & \uletunite {\funtrans \le} \uunite
      \\
      \funtrans {\lboxe \le}
      & \defeq
      & \upaire \uunite {\uprje 2 {\funtrans \le}}
      \\
      \funtrans {\lunboxe \le}
      & \defeq
      & \upaire \uunite {\uprje 2 {\funtrans \le}}
    \end{array}\qquad
    \begin{array}{l@{~}l@{~}l}
      \funtrans \lx
      & \defeq
      & \ux_\lx
      \\
      \funtrans \lloc
      & \defeq
      & \ux_\lloc
      \\
      \funtrans {\conf \lemptystore \le}
      & \defeq
      & \funtrans \le
      \\
      \funtrans {\conf {\stextempty \lstore \lloc} \le}
      & \defeq
      & \subst {\funtrans {\conf \lstore \le}} \uunite {\ux_\lloc}
      \\
      \funtrans {\conf {\stextconf \lstore \lloc \lstorepr \lv} \le}
      & \defeq
      & \subst
          {\funtrans {\conf \lstore \le}}
          {\upaire \uunite {\funtrans {\conf \lstorepr \lv}}}
          {\ux_\lloc}
    \end{array}  
  \end{smathpar}

  \caption{Pure Semantics of Linear State}
  \label{fig:ul:trans}
\end{figure}

The definition of the functional translation of linear contexts, terms,
and types is given in \shortref{fig:ul:trans}. To simplify the
translation of terms and the statement of
\shortref{lem:ul:trans-comp}, we assume that a global injective
mapping is chosen from linear variables $\lx$ and locations $\lloc$ to
ML variables $\ux_\lx$ and $\ux_\lloc$, from linear type variables
$\lalpha$ to ML type variables $\ualpha_\lalpha$, and that the input
terms and types have all bound variables distinct from each other and
their free variables. We extend term translation to contexts
$\funtrans {\ulctxt}$ by extending the translation with
$\funtrans \square \defeq \square$.

\begin{lemma}{Typing}
  \label{lem:ul:trans-typing}
  If $\ulljudg \ulGamma \lstore \le \lty$
\end{lemma}

\begin{lemma}[Compositionality]
  \label{lem:ul:trans-comp}
  $\funtrans {\ulctxt \hw \ue} = {\funtrans \ulctxt} \hw {\funtrans \ue}$.
\end{lemma}

\begin{lemma}[Projection]
  \label{lem:ul:trans-proj}
  If $\ue \in \ullang$ is in the $\ulang$ subset, then
  $\funtrans \ue = \ue$.
\end{lemma}

\begin{theorem}[Termination equivalence]
  \label{thm:ul:trans-equiv}
  The reduction of $\funtrans \ue$ in $\ulang$ terminates if and only
  if the reduction of $\ue$ in $\ullang$ terminates.
\end{theorem}

\begin{proof}
  The translation respects the evaluation structure: a value is
  translated into a value, and a position in the original term is
  reducible if and only if the same position is reducible in the
  translation -- both properties are checked by direct induction, on
  values and evaluation contexts.

  Furthermore, the translation was carefully chosen (especially for
  the store operations) so that there is a redex in the translated
  term if and only if there is a redex in the original term, and the
  reduction of the translation is also the translation of the
  reduction. For example, we have
  \begin{smathpar}
    \begin{array}{ll}
      & \hconf {\stextempty \lstore \lloc} {\lboxe {\lpaire \lloc \lv}}
      \\ \lredexstep
      & \hconf {\stextconf {} \lloc \lstore \lv} \lloc
      \\[10pt]

      & \funtrans {\hconf {\stextempty \lstore \lloc} {\lboxe {\lpaire \lloc \lv}}}
      \\ =
      & \uappe {\funtrans {\lboxe {}}} {\upaire \uunite {\subst \lv \lstore \lstore}}
      \\ =
      & \upaire \uunite {\uprje 2 {\upaire \uunite {\subst \lv \lstore \lstore}}}
      \\ \lredexstep
      &  {\upaire \uunite {\subst \lv \lstore \lstore}}
      \\ =
      & \upaire \uunite {\funtrans {\hconf \lstore \lv}}
      \\ =
      & \funtrans {\hconf {\stextconf {} \lloc \lstore \lv} \lloc}
    \end{array}
  \end{smathpar}
  where $\subst {} \lstore \lstore$ denotes the composed substitution
  $\subst {} {\upaire \uunite {\funtrans {\hconf \lstorepr \lv}}} {\lx_\lloc}$
  for each $\stextconf {} \lloc \lstorepr \lv$ in $\lstore$.
\end{proof}

\begin{theorem}[Full Abstraction]
  \label{thm:ul:full-abstraction}
  The embedding of $\ulang$ into $\ullang$ is fully-abstract:
  \begin{smathpar}
    \ujudg \uGamma {\ue \uctxeq \uepr} \uty

    \implies

    \ulujudg \uGamma {\ue \ulctxeq \uepr} \uty
  \end{smathpar}
\end{theorem}


\begin{proof}
  To show that two $\ulang$ terms $\ue, \uepr$ are contextually
  equivalent in $\ullang$, we are given a context $\ulctxt$ in
  $\ullang$ and must prove that $\ulctxt \hw \ue$, $\ulctxt \hw \ue$ are
  equi-terminating.

  From \fullref{thm:ul:trans-equiv} we know that $\ulctxt \hw \ue$ and
  $\funtrans {\ulctxt \hw \ue}$ are equi-terminating, and from
  \fullref{lem:ul:trans-comp} that $\funtrans {\ulctxt \hw \ue}$ is
  equal to ${\funtrans \ulctxt} \hw {\funtrans \ue}$, which is equal
  to ${\funtrans \ulctxt} \hw \ue$ by
  \fullref{lem:ul:trans-proj}. Similarly, $\ulctxt \hw \uepr$ and
  ${\funtrans \ulctxt} \hw \uepr$ are equi-terminating. Because
  ${\funtrans \ulctxt}$ is a context in $\ulang$, we can use our
  assumption that $\ue \uctxeq \uepr$ to conclude that
  ${\funtrans \ulctxt} \hw \ue$ and ${\funtrans \ulctxt} \hw \uepr$
  are equi-terminating.
\end{proof}

\paragraph{Stability}\label{par:stability}
Given that our proof technique relies on translating $\llang$ back
into $\ulang$, it is stable by extension of the language $\ulang$ --
any extension of $\ulang$ that preserves the reduction behavior of
closed programs, for example adding ML references, preserves the
full-abstraction result.
On the contrary, this technique is \emph{not} stable by extension of
$\llang$, and in fact the result could become false if $\llang$ was
extended with abstraction-breaking features. This explains why we
prove that the embedding of $\ulang$ into $\ullang$ is fully-abstract,
but do not prove any result on the embedding of $\llang$ into
$\ullang$: such a result would be immediately broken by considering
a larger general-purpose language, for example adding ML references.
\endgroup}

\begin{version}{\VeryShort}
\begin{theorem}[Full Abstraction]
  The embedding of $\ulang$ into $\ullang$ is fully-abstract:
  \begin{smathpar}
    \ujudg \uGamma {\ue \uctxeq \uepr} \uty

    \implies

    \ulujudg \uGamma {\ue \ulctxeq \uepr} \uty
  \end{smathpar}
\end{theorem}

\begin{version}{\Appendices}
See the full proof argument in \fullref{appendix:full-abstraction}.
\end{version}
\end{version}

\newcommand{\SectionParametricity}{\begingroup
\section{Multi-language parametricity}
\label{sec:logrel}

We discussed the design choice of manipulating lumps $\llumpty \uty$
of any ML type, not just the type variable that motivates them. In the
presence of polymorphism, this generalization is also an important
design choice to preserve parametricity.

Let us define
$\lfont{id}^\lty(\le) \defeq \lump \lty {(\unlump \lty {\le})}
$,
and consider a polymorphic term of the form
$\uabstre \ualpha {\ULenoconf {\dots \lfont{id}^{\llumpedty \ualpha} \dots}}
$.
The (un)lumping operations on a lumped type such as
$\llumpedty \ualpha$ are just the identity: the lumped value is passed
around unchanged, so $\lfont{id}^{\llumpedty \ualpha}(\lv)$ will
reduce to $\lv$. Now, if we instantiate this polymorphic term with
a ML type $\uty$, it will reduce to a term
$\ULenoconf {\dots \lfont{id}^{\llumpedty \uty} \dots} $
whose unlumping operation is still on a lumped type, so is still exactly the
identity.

On the contrary, if we allowed lumps only on type variables, we would
have to push the lump inside $\uty$, and the (un)lumping operations
would become more complex: if $\uty$ starts with an ML product type
$\upairty \wild \wild$, it would be turned into a shared linear pair
$\lbangty {(\lopairty \wild \wild)}$ by unlumping, and back into an ML
pair by lumping. In general, $\lfont{id}^{\lty}$ may perform deep
$\eta$-expansions of lumped values. The fact that, after instantiation
of the polymorphic term, we get a monomorphic term that has
different (but $\eta$-equivalent) computational behavior would cause
meta-theoretic difficulties; this is the approach that was adopted in
previous work on multi-languages with polymorphism by
\citet*{esop-paper}, and it made some of their proofs using a logical
relations argument substantially more complicated. In the logical 
relation, polymorphism is obtained by allowing each polymorphic
variable to be replaced by two types related by an ``admissible''
relation $\urel$, and the notion of admissibility of this previous
work had to force relations to be compatible with $\eta$-expansion,
which complicates the proofs.

In contrast, our handling of lump types as turning arbitrary types
into blackboxes makes type instantiation obviously parametric. To
formally demonstrate this aspect of our design, we develop
a step-indexed logical relation (\shortref{fig:ul-lr-article}) that proves
that our multi-language satisfies a strong parametricity property that
is not disrupted by the linear sublanguage or the cross-language
boundaries. 

The logical relation is a family of relations indexed by closed
``relational types'' which extend the grammar of $\uty,\lty$, which
include a case for admissible relations ${\urel}$ that is used to
enable parametric arguments.
The step index $j$ in the definitions decreases strictly whenever
related values of a type $\uty$ or $\lty$ are defined in terms of
a non-strictly-smaller type; this happens in the definition of the
relation at recursive types $\lrelV{\lmuty \lalpha \lty}{j}$. Because
the language is non-terminating, our relation does not define an
equivalence but an approximation: two expressions $(\ueone, \ueone)$
are related in $\urelE{\uty}{j}$ if $\ueone$ approximates $\uetwo$:
if $\ueone$ reduces to a value in less than $j$ steps, then $\uetwo$
must reduce to a related value.

The definition of admissible relations $\rel{\utyone}{\utytwo}$, used to
define when $\ulang$ values of polymorphic types are related,
$\urelV{\uforallty \ualpha \uty}{}$, is completely standard, which
demonstrates that our notion of boundaries preserves simple
parametricity reasoning.

Although we have a stateful linear language, the logical relations for
the linear types have more in common with a language with explicit
closures---this is another consequence of the remark in
\fullref{subsec:l:store} that the language can also be interpreted
using a functional semantics. The relations for closed $\llang$ values
and expressions, $\lrelV{\lty}{}$ and $\lrelE{\lty}{}$, are indexed
by a type but do not depend on a store typing: the related values may
have different, non-empty store typings. This allows to relate two
programs that are equivalent but allocate different references in
different ways. Furthermore, since all state is \emph{linear}, we
don't need additional machinery to relate stores, since all the values
in a store owned by a value will be reflected in the value.  For
example, the relation for empty locations
$\lrelV{\lemptyboxty}{}$ relates any two arbitrary (empty)
locations, and the relation for non-empty locations
$\lrelV{\lboxedty{\lty}}{}$ relates possibly-distinct locations
that contain related values.

\begin{figure}
\begin{smathpar}
  \begin{array}{rcl}
    \atom{\uty} & \defeq & \setcomp{{\uv}}{\ujudg{\uempty}{\uv}{\uty}}\\
    \rel{\utyone}{\utytwo} & \defeq & \setcomp{\urel : \mathbb{N} \to \powset{\atom{\utyone}\times \atom{\utytwo}}}
                                      {\forall j \leq j'. ~ \urelat{\urel}{j'} \subset \urelat{\urel}{j}}\\ \\
  \urelV{{\urel}}{j} & \defeq & \urelat{\urel}{j}\\
  \urelV{\ufunty \utyone \utytwo}{j} & \defeq &
  \begin{stackTL}
    \setcomp{\pair{\ufune{\uxone}{(\utyone)_1}{\ueone}}{\ufune{\uxtwo}{(\utyone)_2}{\uetwo}}}{~~
      \begin{stackTL}
        \forall j' \leq j, \inurelV{\uvone}{\uvtwo}{\utyone}{j'}.~
        \inurelE{\subst{\ueone}{\uvone}{\uxone}}{\subst{\uetwo}{\uvtwo}{\uxtwo}}{\utytwo}{j'}}
      \end{stackTL}
  \end{stackTL}\\
  \urelV{\uforallty \ualpha \uty}{j} & \defeq &
  \begin{stackTL}
    \setcomp{\pair{\uabstre \ualpha \uvone}{\uabstre \ualpha \uvtwo}}{\forall \utyone,\utytwo,\urel \in \rel{\utyone}{\utytwo}.~
      \inurelV{\uvone}{\uvtwo}{\subst{\uty}{{\urel}}{\ualpha}}{j}
    }\\
  \end{stackTL}\\ \\

  \lrelV{\lunitty}{j} & \defeq &
    \{\lritem{\lemptystore}{\lunite}{\lemptystore}{\lunite}\} \\
  \lrelV{\lopairty \lty \ltypr}{j} & \defeq &
  \begin{stackTL}
    \setcomp{\lritem {\lstorejoin{\lstoreone}{\lstoreonepr}} {\lpaire \lvone \lvonepr}
                     {\lstorejoin{\lstoretwo}{\lstoretwopr}} {\lpaire \lvtwo \lvtwopr}}{ ~~
      \begin{stackTL}
            \inlrelV \lstoreone   \lvone  \lstoretwo    \lvtwo   \lty   j \wedge
            \inlrelV \lstoreonepr \lvonepr \lstoretwopr \lvtwopr \ltypr j
    }
    \end{stackTL}
  \end{stackTL}\\
  \lrelV{\lofunty{\ltypr}{\lty}}{j} & \defeq &
  \begin{stackTL}
  \setcomp{\lritem{\lstoreone}{\lfune{\lx}{\ltypr}{\leone}}
                  {\lstoretwo}{\lfune{\lx}{\ltypr}{\letwo}}}{ \\ \qquad
    \begin{stackTL}
    \forall j' \leq j, \lstoreonepr, \lstoretwopr,
    \inlrelV{\lstoreonedubpr}{\lvone}{\lstoretwodubpr}{\lvtwo}
            {\ltypr}{j'} . \\ \quad
            \begin{stackTL}
            \lstoreonepr = \lstorejoin{\lstoreone}{\lstoreonedubpr} \wedge
            \lstoretwopr = \lstorejoin{\lstoretwo}{\lstoretwodubpr} \Rightarrow 
            \inlrelE{\lstoreonepr}{\subst{\leone}{\lvone}{\lx}}
                    {\lstoretwopr}{\subst{\letwo}{\lvtwo}{\lx}}
                    {\lty}{j'}
  }
  \end{stackTL}
  \end{stackTL}
  \end{stackTL}
  \\
  \lrelV{\lmuty \lalpha \lty}{j} & \defeq &
  \begin{stackTL}
    \setcomp{\lritem{\lstoreone}{\lfolde {\lmuty \lalpha \lty} \lvone}
                    {\lstoretwo}{\lfolde {\lmuty \lalpha \lty} \lvtwo}}{ ~~
      \begin{stackTL}
        \forall j' < j .
        \inlrelV{\lstoreone}{\lvone}
                {\lstoretwo}{\lvtwo}
                {\lunfoldedty{\lalpha}{\lty}}
                {j'}
    }
      \end{stackTL}
  \end{stackTL} \\
  \lrelV{\lbangty{\lty}}{j} & \defeq &
  \begin{stackTL}
    \setcomp{\lritem{\lemptystore}{\lsharee{\lstoreone}{\lstoretyone}{\lvone}}
                    {\lemptystore}{\lsharee{\lstoretwo}{\lstoretytwo}{\lvtwo}}}{ ~~
      \begin{stackTL}
        \inlrelV{\lstoreone}{\lvone}{\lstoretwo}{\lvtwo}{\lty}{j}
      }
      \end{stackTL}
  \end{stackTL} \\
  \lrelV{\lemptyboxty}{j} & \defeq &
  \{\lritem{\stextempty {} \llocone}{\llocone}
           {\stextempty {} \lloctwo}{\lloctwo}\} \\
  \lrelV{\lboxedty \lty}{j} & \defeq &
  \begin{stackTL}
  \setcomp{\lritem{\stextconf {} \llocone \lstoreone \lvone}{\llocone}
    {\stextconf {} \lloctwo \lstoretwo \lvtwo}{\lloctwo}}{ ~~
    \begin{stackTL}
      \inlrelV{\lstoreone}{\lvone}{\lstoretwo}{\lvtwo}{\lty}{j}
    }
    \end{stackTL}
  \end{stackTL} \\

  \lrelV{\llumpty{\uty}}{j} & \defeq &
  \begin{stackTL}
  \setcomp{\lritem{\lemptystore}{\llumpe{\uvone}}{\lemptystore}{\llumpe\uvtwo}}
  {\inurelV{\uvone}{\uvtwo}{\uty}{j}}
  \end{stackTL}
  \\ \\

  \urelE{\uty}{j} & \defeq &
  \begin{stackTL}
    \setcomp{\pair{\ueone}{\uetwo}}{
      \begin{stackTL}
        \forall j' \leq j.~
        \ueone \uredexstep^{j'} \uvone \Rightarrow
        \exists \uvtwo.~ \uetwo \uredexstep^{*} \uvtwo \wedge
        \inurelV{\uvone}{\uvtwo}{\uty}{j-j'}
      }
      \end{stackTL}
  \end{stackTL}\\
  \lrelE{\lty}{j} & \defeq &
  \begin{stackTL}
    \setcomp{\lritem{\lstoreone}{\leone}{\lstoretwo}{\letwo}}{ ~~
      \begin{stackTL}
        \forall j' \leq j, \hconf \lstoreonepr \lvone .
          \hconf \lstoreone \leone \lredexstepin{j'} \hconf \lstoreonepr \lvone \Rightarrow
          \\ \quad
          \begin{stackTL}
          \exists \hconf \lstoretwopr \lvtwo .
            \hconf \lstoretwo \letwo \lredexstepstar \hconf \lstoretwopr \lvtwo ~\wedge~ 
            \inlrelV{\lstoreonepr}{\lvone}{\lstoretwopr}{\lvtwo}{\lty}{j-j'}
      }
          \end{stackTL}
      \end{stackTL}
  \end{stackTL}\\ \\
  \end{array}
\end{smathpar}
\caption{Multi-language Logical Relation (excerpt)}
\label{fig:ul-lr-article}
\end{figure}

Logical relations effectively translate global invariants of the
system into properties of type connectives. For example, consider the
reduction rule for lumped values:
%
$
  \ULe \lemptystore \lemptystorety {\lsrcsharee {\llumpe \uv}} \uredexstep \uv
$.
%
In this rule we implicitly assumed that a linear value of the shape
$\llumpe \uv$ at type $\lbangty {\llumpty \uty}$ would occur in an
empty local store. The term $\lsrcsharee {\llumpe \uv}$ desugars into
$\lsharee \lemptystore \lemptystorety {\llumpe \uv}$, but it is not
immediately obvious that this should always be the case since it is
possible to compute a value of type $\lbangty {\llumpty \uty}$ by
allocating references and using them. The intuitive reason why the
store becomes empty when a value $\llumpe \uv$ is reached is that
linear sub-terms $\le$ within $\uv$ may only occur within a language
boundary $\ULe \lstore \lstorety \le$: linear sub-terms have their own 
local store, so there are no globally visible linear locations for
$\llumpty \uv$ to refer to. This global reasoning is elegantly
expressed in a type-directed way in our logical relation by the
definition of related values at lump type, which encodes the invariant
that they always have an empty store:
\begin{mathpar}
  \lrelV{\llumpty{\uty}}{j} \defeq
  \begin{stackTL}
  \setcomp{\lritem{\lemptystore}{\llumpe{\uvone}}{\lemptystore}{\llumpe\uvtwo}}
  {\inurelV{\uvone}{\uvtwo}{\uty}{j}}
  \end{stackTL}
\end{mathpar}

From the logical relation defined on closed $\ullang$ terms and
values, we define logical approximation relations on open terms
$\uapproxjdg \ulGamma \ueone \uetwo \uty$ and
$\Lapproxjdg \ulGamma \lstoreone \leone \lstoretwo \letwo \lty$ by
asking for the open terms and values to be related under all related
environments in the standard way---see Appendix~\ref{ann:logrel} for
full details. This lets us demonstrate the Fundamental Property to validate
the construction of our logical relation, showing that all typing
rules are admissible.

\begin{theorem}[Fundamental Property]{~}
  \begin{enumerate}
  \item If $\ujudg{\ulbGamma}{\ue}{\uty}$
    then $\uapproxjdg{\ulbGamma}{\ue}{\ue}{\uty}$
  \item If $\ljudg{\lstorety}{\ulGamma}{\lstore}{\le}{\lty}$
    then $\Lapproxjdg{\ulGamma}{\lstore}{\le}{\lstore}{\le}{\lty}$
  \end{enumerate}
\end{theorem}

We also prove that the logical relation is sound with respect to
contextual equivalence. For this, we define contextual approximation
relations $\ucapproxjdg \ulGamma \ueone \uetwo \uty$ and
$\Lcapproxjdg \ulGamma \lstoreone \leone \lstoretwo \letwo \lty$,
by asking that $\ueone, \conf \lstoreone \leone$ terminate more often
than $\uetwo, \conf \lstoretwo \letwo$ when run under arbitrary
contexts---see Appendix~\ref{ann:logrel} for details.

\begin{theorem}[Soundness of Logical Relation with respect to
  Contextual Equivalence]
  $(\lapproxsymbol) \subset (\capproxsymbol)$.
\end{theorem}
\endgroup}

\Not\VeryShort{\SectionParametricity}{}

\Not\VeryShort{\section{Hybrid program examples}
\label{sec:examples}

\subsection{In-Place Transformations}

In \fullref{subsec:l:store} we proposed a program for in-place reversal
of linear lists defined by the type
$
\lapptyop {LinList} \lty
\defeq
\lmuty \lalpha
  {\losumty
    \lunitty
    {\lboxedty {(\lopairty \lty \lalpha)}}}
$.
We can also define a type of ML lists
$\uapptyop {List} \uty
\defeq
\umuty \ualpha {\usumty \uunitty {\upairty \uty \ualpha}}
$.
Note that ML lists are compatible with shared linear lists,
in the sense that
$\uapptyop {List} \uty
\ulcompatbangpar
{\lapptyop {LinList} {\llumpedty \uty}}
$.
This enables writing in-place list-manipulation functions in $\llang$,
and exposing them to beginners at a $\ulang$ type:
\begin{smathpar}
  \ufont{rev}~\ufont{xs}
  \defeq
  \UmpLenoconf
    {\lapptyop {LinList} {\llumpedty \uty}}
    {\lsrcsharee
      {(\lfont{rev\_into}
        ~{\lsrccopye {(\LUmpe {\lapptyop {LinList} {\llumpedty \uty}} {\ufont{xs}})}}
        ~\lfont{Nil})}}
\end{smathpar}

This example is arguably silly, as the allocations that are avoided by
doing an in-place traversal are paid when copying the shared list to
obtain a uniquely-owned version. A better example of list operations
that can profitably be sent on the linear side is quicksort, whose
code we give in \fullref{fig:quickort}. An ML implementation allocates
intermediary lists for each recursive call, while the surprisingly
readable $\ulang$ implementation only allocates for the first copy.

\begin{figure}
  \centering
\begin{lstlisting}
/*
partition : !(!$\lalpha$ $\multimap$ Bool) $\multimap$ LList !$\lalpha$ $\multimap$ LList !$\lalpha$ $\oplus$ LList !$\lalpha$
partition p li = partition_aux p (Nil, Nil) li
partition_aux p (yes, no) Nil = (yes, no)
partition_aux p (yes, no) (Cons l x xs) =
  let (yes, no) =
    if copy p x
    then (Cons l x yes, no)
    else (yes, Cons l x no) in
  partition_aux p (yes, no) xs

lin_quicksort : LList !$\lalpha$ $\multimap$ LList !$\lalpha$
lin_quicksort li = quicksort_aux li Nil
quicksort_aux Nil acc = acc
quicksort_aux (Cons l head li) acc =
  let p = share (fun x -> x < head) in
  let (below, above) = partition p li in
  quicksort_aux below (Cons l head (quicksort_aux above acc))*/

/!quicksort li UL(/*li*/) = UL(/*share (lin_quicksort (copy li))*/)!/
\end{lstlisting}
  \caption{Quicksort}
  \label{fig:quickort}
\end{figure}

\subsection{Typestate Protocols}

Linear types can enforce proper allocation and deallocation of
resources, and in general any automata/typestate-like protocols on
their usage by encoding the state transitions as linear
transformations. In the simple example of file-descriptor handling in
the introduction, additional safety compared to ML programming can be
obtained by exposing file-handling functions on the $\ulang$ side,
with linear types. We assumed the following API for linear file handling, which enforces a correct usage protocol:
\newcommand{\lPath}{\llumpedty{\ufont{Path}}}
\newcommand{\lString}{\llumpedty{\ufont{String}}}
\newcommand{\lHandle}{\lfont{Handle}}
\newcommand{\lEmptyHandle}{\lfont{EmptyHandle}}
\begin{smathpar}
  \begin{array}{lll}
    \lfont{open} & :
    & \lbangty{(\lofunty \lPath \lHandle)}
    \\
    \lfont{line} & :
    & \lbangty
        {(\lofunty \lHandle
          {(\losumty \lHandle {(\lopairty \lString \lHandle)})})}
    \\
    \lfont{close} & :
    & \lbangty{(\lofunty \lHandle \lunitty)}
  \end{array}
\end{smathpar}

Another interesting example of protocol usage for which linear types help
is the use of \emph{transient} versions of persistent data structures,
as popularized by Clojure. An unrestricted type
$\uapptyop {Set} \ualpha$ may represent persistent sets as balanced
trees with logarithmic operations performing
path-copying. A $\lfont{transient}$ call returns a mutable version of
the structure that supports efficient batch in-place updates, before
a $\lfont{persistent}$ call freezes this transient structure back into
a persistent tree. To preserve a purely functional semantics, we must
enforce that the intermediate transient value is uniquely owned. We
can do this by using the linear types for the transient API:

\begin{small}
\begin{lstlisting}
/!type Set $\ualpha$
val add : Set $\ualpha$ $\ufunty{}{}$ $\ualpha$ $\ufunty{}{}$ Set $\ualpha$ (* path copy *)
...!/
/*type MutSet $\lalpha$
val add: !(MutSet $\lalpha$ $\lofunty{}{}$ $\lalpha$ $\lofunty{}{}$ MutSet $\lalpha$) (* in-place update *)
...
val transient : !(![/!Set $\ualpha$!/] $\lofunty{}{}$ MutSet ![/!$\ualpha$!/])
val persistent : !(MutSet ![/!$\ualpha$!/] $\lofunty{}{}$ ![/!Set $\ualpha$!/])*/
\end{lstlisting}
\end{small}
}{}

\begin{version}{\False}
\section{Implementation}
\label{sec:implementation}

We have a prototype implementation for $\ullang$. In this
implementation, instead of $\ulang$ we use the full OCaml language; we
implemented the linear language $\llang$ and wrote a compiler from
$\llang$ to (unsafe) OCaml code using unsafe in-place mutation.
Here is what in-place reversal looks like in the syntax supported by
our prototype: 
\begin{small}
\begin{lstlisting}
/!/*(%%L
  type llist 'a = Nil | Cons of box1 ('a * llist 'a)
  let rec rev_append = fun li : llist 'a -> fun acc : llist 'a -> match li with
    | Nil -> acc
    | Cons (x, xs)@l -> rev_append xs (Cons (x, acc)@l)
)*/
type 'a list = Nil | Cons of ('a * 'a list)
let rev (li : 'a list) = /*(%L share (rev_append (copy /!(%U li :> /*!(llist 'a)*/)!/) Nil))*/!/
\end{lstlisting}
\end{small}

The syntax for the boundaries is \texttt{(\%L ..)} and \texttt{(\%U
  ..)} in expressions and patterns, and \texttt{(\%\%L ..)} and
\texttt{(\%\%U ..)} for sequences of toplevel declarations. The
U parts accept the full grammar of the OCaml programming language
(version 4.04.0), and use the OCaml implementation for type-checking
and compilation. The L parts use our own parser and type-checker that
enforces the linear discipline; in particular, we have not implemented
type inference, so function parameters are fully specified, and
U boundaries within L terms come with an annotation
(\texttt{:> !(llist 'a)} in this example) indicating at which L terms
they should be unlumped.

Finally, \texttt{(x, xs)@l} is syntactic sugar for boxing and
unboxing, available in both $\lsymbol$ expressions and patterns. In an
expression, it is equivalent to \texttt{box (l, (x, xs))}. In
a pattern, it unboxes the reference \texttt{l} and matches its
contents with the pattern \texttt{(x, xs)}.

As we previously pointed out, our proof technique for
\fullref{thm:ul:full-abstraction} is stable by extension of the
general-purpose language, as long as the extended language keeps
reducing closed $\ulang$ programs in the same way---we admit that
this is the case for OCaml. This means that using OCaml as the
general-purpose language does not endanger the full-abstraction
result: we have formally established that $\llang$ does not leak into
OCaml's abstractions. OCaml programmers can now use our prototype to
safely add resource control or in-place update to their programs.
\end{version}

\begin{version}{\Not\VeryShort}
\section{Conclusion}

In our proposed multi-language design, a simple linear type system
mirroring the standard rules of intuitionistic linear logic can be
equipped with linear state and usefully complement a general-purpose
functional ML language, without breaking equational reasoning or
parametricity---and without requiring a significantly more complex
meta-theory.

Fine-grained language boundaries allow interesting programming
patterns to emerge, and full abstraction provides a novel rigorous
specification of what it means for multi-language design to avoid
\emph{abstraction leaks} from advanced features into the
general-purpose or beginner-friendly languages.

\subsection{Related Work}
\end{version}

\begin{version}{\VeryShort}
\section{Conclusion and Related Work}
\end{version}

Having a stack of usable, interoperable languages, extensions or
dialects is at the forefront of the Racket approach to programming 
environments, in particular for
teaching~\citep*{felleisen2004teachscheme}. 

Our multi-language semantics builds on the seminal work 
by~\citet*{matthews2009operational}, who gave a formal semantics of 
interoperability between a dynamically and a statically typed language.
Others have followed the Matthews-Findler 
approach of designing multi-language systems with fine-grained
boundaries---for instance, formalizing interoperability between a simply 
and dependently typed language~\citep*{osera12}; between a functional
and typed assembly language~\citep*{patterson17:funtal}; between an
ML-like and an affinely typed language, where linearity is enforced at
runtime on the ML side using stateful contracts~\citep*{tov10}; and between
the source and target languages of compilation to specify compiler
correctness~\citep*{esop-paper}.  However, all these papers address
only the question of soundness of the multi-language; we propose
a formal treatment of \emph{usability} and absence of abstraction
leaks.  

The only work to establish that a language embeds into a
multi-language in a fully abstract way is the work on fully abstract
compilation by~\citet*{ahmed11:epcps} and~\citet*{new16:facue} who
show that their compiler's source language embeds into their source-target
multi-language in a fully abstract way.  But the focus of this work
was on fully abstract compilation, not on usability of user-facing
languages. 

The Eco project~\citep*{eco} is studying multi-language systems where
user-exposed languages are combined in a very fine-grained way; it is
closely related in that it studies the user experience in
a multi-language system. The choice of an existing dynamic language
creates delicate interoperability issues (conflicting variable scoping
rules, etc.) as well as performance challenges. We propose a different
approach, to design new multi-languages from scratch with
interoperability in mind to avoid legacy obstacles.




We are not aware of existing systems exploiting the simple idea of
using promotion to capture uniquely-owned state and dereliction to
copy it---common formulations would rather perform copies on the
contraction rule.

The general idea that linear types can permit reuse of unused
allocated cells is not new. In \citet*{wadler}, a system is proposed
with both linear and non-linear types to attack precisely this
problem. It is however more distant from standard linear logic and
somewhat ad-hoc; for example, there is no way to permanently turn
a uniquely-owned value into a shared value, it provides instead
a local \emph{borrowing} construction that comes with ad-hoc
restrictions necessary for safety.
(The inability to \emph{give up} unique ownership, which is essential
in our list-programming examples, seems to also be missing from Rust,
where one would need to perform a costly operation of traversing the
graph of the value to turn all pointers into \texttt{Arc} nodes.)

The RAML project~\citep*{HAH12} also combines linear logic and memory
reuse: its \emph{destructive match} operator will implicitly reuse
consumed cells in new allocations occurring within the match
body. Multi-languages give us the option to explore more explicit,
flexible representations of those low-level concern, without imposing
the complexity to all programmers.

A recent related work is the Cogent language~\citep*{cogent}, in which
linear state is also viewed as both functional and imperative -- the
latter view enabling memory reuse. The language design is
interestingly reversed: in Cogent, the linear layer is the simple
language that everyone uses, and the non-linear language is a complex
but powerful language that is used when one really has to, named C.


Our linear language $\llang$ is sensibly simpler, and in several ways
less expressive, than advanced programming languages based on linear
logic~\citep*{alms}, separation logic~\citep*{mezzo}, fine-grained
permissions~\citep*{plaid}: it is not designed to stand on its own,
but to serve as a useful side-kick to a functional language, allowing
safer resource handling.

One major simplification of our design compared to more advanced
linear or separation-logic-based languages is that we do not separate
physical locations from the logical capability/permission to access
them (e.g., as in~\citet*{ahmed07:L3}). This restricts expressiveness
in well-understood ways~\citep*{adoption-focus}: shared values cannot
point to linear values.

Alms~\citep*{alms}, Quill~\citep*{morris16} and Linear Haskell
\citep*{linear-haskell} add linear types to a functional language, trying
hard not to lose desirable usability property, such as type inference
or the genericity of polymorphic higher-order functions. This is very
challenging; for example, Linear Haskell gives up on principality of
inference\footnote{Thanks to Stephen Dolan for pointing out that
  $\lambda f. \lambda x.\,f~x$ has several incompatible
Linear Haskell types.}. Our multi-language design side-steps this
issue as the general-purpose language remains unchanged. Language
boundaries are more rigid than an ideal no-compromise language, as
they force users to preserve the distinction between the
general-purpose and the advanced features; it is precisely this
compromise that gives a design of reduced complexity.

Finally, on the side of the semantics, our system is related to
LNL~\citep*{LNL}, a calculus for linear logic that, in a sense, is
itself built as a multi-language system where (non-duplicable) linear
types and (duplicable) intuitionistic types interact through
a boundary. It is not surprising that our design contains an
instance of this adjunction: for any $\lty$ there is a unique $\uty$
such that $\uty \ulcompat \lbangty \lty$, and converting a $\lty$
value to this $\uty$ and back gives a $\lbangty \lty$ and is provably
equivalent, by boundary cancellation, to just using $\lsrcsharee {}$.


\subsection*{Acknowledgments}

We thank our anonymous reviewers for their feedback, as well as
Neelakantan Krishnaswami, François Pottier, Jennifer Paykin, Sylvie
Boldo and Simon Peyton-Jones for our discussions on this work.

This work was supported in part by the National Science Foundation
under grants CCF-1422133 and CCF-1453796, and the European Research
Council under ERC Starting Grant SECOMP (715753). Any opinions,
findings, and conclusions expressed in this material are those of the
authors and do not necessarily reflect the views of our funding
agencies.

\newpage
\bibliographystyle{plainnat}
\bibliography{article,amal}

\begin{version}{\Appendices}
\newpage
\appendix

\begin{version}{\VeryShort}
\section{Internal \texorpdfstring{$\llang$}{L} Syntax and Typing}
\label{appendix:internal-linear-language}

\SectionInternalLL
\end{version}

\begin{version}{\VeryShort}
\section{Reduction of Internal Terms}
\label{appendix:internal-linear-reduction}

\SectionInternalLLReduction
\end{version}

\begin{version}{\VeryShort}
\section{Interoperability: Static Semantics}
\label{appendix:interop-static}

\SectionInteropStatic
\end{version}

\begin{version}{\VeryShort}
\section{Interoperability: Dynamic Semantics}
\label{appendix:interop-dynamic}

\SectionInteropDynamic
\end{version}

\begin{version}{\VeryShort}
\section{Full Abstraction from \texorpdfstring{$\ulang$}{U}
  into \texorpdfstring{$\ullang$}{UL}}
\label{appendix:full-abstraction}

\SectionFullAbstraction
\end{version}

\VeryShort{\SectionParametricity}{}

\VeryShort{}{}

\section{Logical relation}
\label{ann:logrel}

To define the logical relation in a precise way, we introduce a new
grammar of ``relational types'' $\urty,\lrty$, that simply extend the
grammar of $\uty,\lty$ to include a case for a relation on $\ulang$
types $\triple{\urel}{\utyone}{\utytwo}$.

Introducing this lightweight syntax for relations here is a middle way
between definitions that use an explicit relational substitution
(cluttering the relation with yet another index) and a full-blown
logic for parametricity as in \citet{plotkin1993logic}

\begin{figure}
  \begin{displaymath}
    \begin{array}{crl}
      \urty
      & \bnfdef
      & \ldots \text{all cases for $\uty$, but using $\urty,\lrty$ recursively} \bnfalt \triple{\urel}{\utyone}{\utytwo}
      \\[3pt]

      \lrty
      & \bnfdef
      & \text{all cases for $\lty$, but using $\urty,\lrty$ recursively}
    \end{array}
  \end{displaymath}
  \caption{Relation Type Syntax}
  \label{fig:reltype:syntax}
\end{figure}

Every $\urty$ has two associated types, the types of terms that it
relates, which we denote $\ureldomone{\urty}, \ureldomtwo{\urty}$. It is
defined as follows:

\begin{smathpar}
  \begin{array}{rcl}
  \ureldomone{\triple{\urel}{\utyone}{\utytwo}} & \defeq & \utyone\\
  \ureldomtwo{\triple{\urel}{\utyone}{\utytwo}} & \defeq & \utytwo\\
  \ureldomi{\ualpha} & \defeq & \ualpha \\
  \ureldomi{\upairty \urtyone \urtytwo} & \defeq & \upairty{\ureldomi{\urtyone}}{\ureldomi{\urtytwo}}\\
  \ureldomi{\uunitty} & \defeq & {\uunitty}\\
  \ureldomi{\ufunty \urtyone \urtytwo} & \defeq & {\ufunty{\ureldomi{\urtyone}}{\ureldomi{\urtytwo}}}\\
  \ureldomi{\usumty \urtyone \urtytwo} & \defeq & {\usumty {\ureldomi{\urtyone}} {\ureldomi{\urtytwo}}}\\
  \ureldomi{\umuty \ualpha \urty} & \defeq & {\umuty \ualpha {\ureldomi{\urty}}}\\
  \ureldomi{\uforallty \ualpha \urty} & \defeq & {\uforallty \ualpha {\ureldomi{\urty}}}\\
  \end{array}

  \begin{array}{rcl}
  \ureldomi{\lopairty \lrtyone \lrtytwo} & \defeq & {\lopairty \lrtyone \lrtytwo}\\
  \ureldomi{\lunitty} & \defeq & {\lunitty}\\
  \ureldomi{\lofunty \lrtyone \lrtytwo} & \defeq & {\lofunty {\ureldomi{\lrtyone}} {\ureldomi{\lrtytwo}}}\\
  \ureldomi{\losumty \lrtyone \lrtytwo} & \defeq & {\losumty {\ureldomi{\lrtyone}} {\ureldomi{\lrtytwo}}}\\
  \ureldomi{\lmuty \lalpha \lrty} & \defeq & {\lmuty \lalpha {\ureldomi{\lrty}}}\\
  \ureldomi{\lalpha} & \defeq & {\lalpha}\\
  \ureldomi{\lbangty \lrty} & \defeq & {\lbangty {\ureldomi{\lrty}}}\\
  \ureldomi{\lboxedty \lrty} & \defeq & {\lboxedty {\ureldomi{\lrty}}}\\
  \ureldomi{\lemptyboxty} & \defeq & {\lemptyboxty}\\
  \end{array}
\end{smathpar}

First, we define when closed values are related at each type, indexing
by a natural number to break the circularity of recursive types. The
relations are defined by nested induction on $j,\urty,\lrty$, any time
a bigger type is used in a definition, the step-index $j$ is
decremented.

The definition of $\urelE{\uty}{j}$ shows that this is an assymetric
relation capturing a notion of \emph{approximation}, not equivalence.

\begin{figure}
  \begin{smathpar}
    \begin{array}{rcl}
    \atom{\uty} & \defeq & \setcomp{{\uv}}{\ujudg{\uempty}{\uv}{\uty}}\\
    \rel{\utyone}{\utytwo} & \defeq & \setcomp{\urel : \mathbb{N} \to \powset{\atom{\utyone}\times \atom{\utytwo}}}
                                      {\forall j \leq j'. ~ \urelat{\urel}{j'} \subset \urelat{\urel}{j}}\\ \\

  \urelV{\triple{\urel}{\utyone}{\utytwo}}{j} & \defeq & \urelat{\urel}{j}\\
  \urelV{\uunitty}{j} & \defeq & \{{\pair{\uunite}{\uunite}}\}\\
  \urelV{\upairty \urty \urtypr}{j} & \defeq &
    \setcomp{\pair{\upaire{\uvone}{\uvonepr}}{\upaire{\uvtwo}{\uvtwopr}}}
            {\inurelV{\uvone}{\uvtwo}{\urty}{j} \wedge
             \inurelV{\uvone}{\uvtwo}{\urty}{j}}\\
  \urelV{\usumty \urtyone \urtytwo}{j} & \defeq &
    \setcomp{\pair{\usume{i}{\uvone}}{\usume{i}{\uvtwo}}}
            {\inurelV{\uvone}{\uvtwo}{\urtyi}{j}}\\
  \urelV{\umuty \ualpha \urty}{j} & \defeq &
    \setcomp{\pair{\ufolde{({\umuty \ualpha \urty})_1}{\uvone}}{\ufolde{({\umuty \ualpha \urty})_2}{\uvone}}}{\forall j' < j.~ \inurelV{\uvone}{\uvtwo}{\uunfoldedty{\ualpha}{\urty}}{j'}}\\
  \urelV{\ufunty \urtyone \urtytwo}{j} & \defeq &
  \begin{stackTL}
    \setcomp{\pair{\ufune{\uxone}{(\urtyone)_1}{\ueone}}{\ufune{\uxtwo}{(\urtyone)_2}{\uetwo}}}{\\\qquad
      \begin{stackTL}
        \forall j' \leq j, \inurelV{\uvone}{\uvtwo}{\urtyone}{j'}.~
        \inurelE{\subst{\ueone}{\uvone}{\uxone}}{\subst{\uetwo}{\uvtwo}{\uxtwo}}{\urtytwo}{j'}}
      \end{stackTL}
  \end{stackTL}\\
  \urelV{\uforallty \ualpha \urty}{j} & \defeq &
  \begin{stackTL}
    \setcomp{\pair{\uabstre \ualpha \uvone}{\uabstre \ualpha \uvtwo}}{\forall \utyone,\utytwo,\urel \in \rel{\utyone}{\utytwo}.~
      \inurelV{\uvone}{\uvtwo}{\subst{\urty}{\triple{\urel}{\utyone}{\utytwo}}{\ualpha}}{j}
    }\\
  \end{stackTL}\\
  \lrelV{\lunitty}{j} & \defeq &
    \{\lritem{\lemptystore}{\lunite}{\lemptystore}{\lunite}\} \\
  \lrelV{\lopairty \lrty \lrtypr}{j} & \defeq &
  \begin{stackTL}
    \setcomp{\lritem {\lstorejoin{\lstoreone}{\lstoreonepr}} {\lpaire \lvone \lvonepr}
                     {\lstorejoin{\lstoretwo}{\lstoretwopr}} {\lpaire \lvtwo \lvtwopr}}{ \\ \qquad
      \begin{stackTL}
            \inlrelV \lstoreone   \lvone  \lstoretwo    \lvtwo   \lrty   j \wedge
            \inlrelV \lstoreonepr \lvonepr \lstoretwopr \lvtwopr \lrtypr j
    }
    \end{stackTL}
  \end{stackTL}\\
  \lrelV{\losumty{\lrtyone}{\lrtytwo}}{j} & \defeq &
  \begin{stackTL}
    \setcomp{\lritem{\lstoreone}{\lsume i \lvone}{\lstoretwo}{\lsume i \lvtwo}}{ \\ \qquad
      \inlrelV{\lstoreone}{\lvone}{\lstoretwo}{\lvtwo}{\lrtyi}{j}
    }
  \end{stackTL}\\
  \lrelV{\lmuty \lalpha \lrty}{j} & \defeq &
  \begin{stackTL}
    \setcomp{\lritem{\lstoreone}{\lfolde {\lmuty \lalpha \lrty} \lvone}
                    {\lstoretwo}{\lfolde {\lmuty \lalpha \lrty} \lvtwo}}{ \\ \qquad
      \begin{stackTL}
        \forall j' < j .
        \inlrelV{\lstoreone}{\lvone}
                {\lstoretwo}{\lvtwo}
                {\lunfoldedty{\lalpha}{\lrty}}
                {j'}
    }
      \end{stackTL}
  \end{stackTL} \\
  \lrelV{\lofunty{\lrtypr}{\lrty}}{j} & \defeq &
  \begin{stackTL}
  \setcomp{\lritem{\lstoreone}{\lfune{\lx}{\lrtypr}{\leone}}
                  {\lstoretwo}{\lfune{\lx}{\lrtypr}{\letwo}}}{ \\ \qquad
    \begin{stackTL}
    \forall j' \leq j, \lstoreonepr, \lstoretwopr,
    \inlrelV{\lstoreonedubpr}{\lvone}{\lstoretwodubpr}{\lvtwo}
            {\lrtypr}{j'} . \\ \quad
            \begin{stackTL}
            \lstoreonepr = \lstorejoin{\lstoreone}{\lstoreonedubpr} \wedge
            \lstoretwopr = \lstorejoin{\lstoretwo}{\lstoretwodubpr} \Rightarrow \\
            \inlrelE{\lstoreonepr}{\subst{\leone}{\lvone}{\lx}}
                    {\lstoretwopr}{\subst{\letwo}{\lvtwo}{\lx}}
                    {\lrty}{j'}
  }
  \end{stackTL}
  \end{stackTL}
  \end{stackTL}
  \\
  \lrelV{\lbangty{\lrty}}{j} & \defeq &
  \begin{stackTL}
    \setcomp{\lritem{\lemptystore}{\lsharee{\lstoreone}{\lstoretyone}{\lvone}}
                    {\lemptystore}{\lsharee{\lstoretwo}{\lstoretytwo}{\lvtwo}}}{ \\ \qquad
      \begin{stackTL}
        \inlrelV{\lstoreone}{\lvone}{\lstoretwo}{\lvtwo}{\lrty}{j}
      }
      \end{stackTL}
  \end{stackTL} \\
  \lrelV{\lemptyboxty}{j} & \defeq &
  \{\lritem{\stextempty {} \llocone}{\llocone}
           {\stextempty {} \lloctwo}{\lloctwo}\} \\
  \lrelV{\lboxedty \lrty}{j} & \defeq &
  \begin{stackTL}
  \setcomp{\lritem{\stextconf {} \llocone \lstoreone \lvone}{\llocone}
    {\stextconf {} \lloctwo \lstoretwo \lvtwo}{\lloctwo}}{ \\ \qquad
    \begin{stackTL}
      \inlrelV{\lstoreone}{\lvone}{\lstoretwo}{\lvtwo}{\lrty}{j}
    }
    \end{stackTL}
  \end{stackTL} \\

  \lrelV{\llumpty{\urty}}{j} & \defeq &
  \begin{stackTL}
  \setcomp{\lritem{\lemptystore}{\llumpe{\uvone}}{\lemptystore}{\llumpe\uvtwo}}
  {\inurelV{\uvone}{\uvtwo}{\urty}{j}}
  \end{stackTL}
  \\ \\

  \urelE{\urty}{j} & \defeq &
  \begin{stackTL}
    \setcomp{\pair{\ueone}{\uetwo}}{
      \begin{stackTL}
        \forall j' \leq j.~
        \ueone \uredexstep^{j'} \uvone \Rightarrow\\ \quad
        \exists \uvtwo.~ \uetwo \uredexstepstar \uvtwo \wedge
        \inurelV{\uvone}{\uvtwo}{\urty}{j-j'}
      }
      \end{stackTL}
  \end{stackTL} \\
  \lrelE{\lrty}{j} & \defeq &
  \begin{stackTL}
    \setcomp{\lritem{\lstoreone}{\leone}{\lstoretwo}{\letwo}}{ \\ \qquad
      \begin{stackTL}
        \forall j' \leq j, \hconf \lstoreonepr \lvone .
          \hconf \lstoreone \leone \lredexstepin{j'} \hconf \lstoreonepr \lvone \Rightarrow
          \\ \quad
          \begin{stackTL}
          \exists \hconf \lstoretwopr \lvtwo .
            \hconf \lstoretwo \letwo \lredexstepstar \hconf \lstoretwopr \lvtwo \wedge \\ \quad
            \inlrelV{\lstoreonepr}{\lvone}{\lstoretwopr}{\lvtwo}{\lrty}{j-j'}
      }
          \end{stackTL}
      \end{stackTL}
  \end{stackTL}
    \end{array}
  \end{smathpar}
  \caption{Logical Approximation for Closed Terms}
\end{figure}

Next, we extend the relations to open terms by defining open terms to
be related when they are related when closed by related substitutions.

\begin{figure}
  \begin{smathpar}
  \begin{array}{rcl}
  \lrelG{\lemptyGamma}{j} & \defeq & \{\greltriple{\lemptystore}{\lemptystore}{\lemptysubsts}\}\\

  \lrelG{\ulGamma, \var \lx \lty}{j} & \defeq &
  \begin{stackTL}
    \setcomp{\greltriple{\lstorejoin{\lstoreone}{\lstoreonepr}}
                    {\lstorejoin{\lstoretwo}{\lstoretwopr}}
                    {\mapext{\ulsubsts}{\lx}{\pair{\lvone}{\lvtwo}}}}
    {\\ \qquad
        \inlrelG{\lstoreone}{\lstoretwo}{\ulsubsts}{\ulGamma}{j} \wedge
        \inlrelV{\lstoreonepr}{\lvone}{\lstoretwopr}{\lvtwo}{{\lappsubstsrel{\ulsubsts}{\lty}}}{j}
    }
  \end{stackTL}\\
  \lrelG{\ulGamma, \var{\ux}{\uty}}{j} & \defeq &
    \begin{stackTL}
      \setcomp{\greltriple{\lstoreone}{\lstoretwo}{\mapext{\ulsubsts}{\ux}{\pair{\uvone}{\uvtwo}}}
      }{\\\qquad
        \inlrelG{\lstoreone}{\lstoretwo}{\ulsubsts}{\ulGamma}{j} \wedge
        \inurelV{\uvone}{\uvtwo}{\lappsubstsrel{\ulsubsts}{\uty}}{j}
      }
    \end{stackTL} \\
  \lrelG{\ulGamma, \ualpha}{j} & \defeq &
    \begin{stackTL}
      \setcomp{\greltriple{\lstoreone}{\lstoretwo}{\mapext{\ulsubsts}{\ualpha}{\triple{\urel}{\utyone}{\utytwo}}}
      }{\\\qquad
        {{\urel}} \in \rel{\utyone}{\utytwo}\wedge
        \inlrelG{\lstoreone}{\lstoretwo}{\ulsubsts}{\ulGamma}{j}
      }
    \end{stackTL}
  \end{array}
  \end{smathpar}
  \begin{smathpar}
    \uapproxjdg{\ulbGamma}{\ueone}{\uetwo}{\uty}\defeq
    \begin{stackTL}
      \forall j \geq 0,
      \inlrelG{\lemptystore}{\lemptystore}{\ulsubsts}{\ulbGamma}{j}.
      \inurelE{\lappsubstsone{\ulsubsts}{\ueone}}
      {\lappsubststwo{\ulsubsts}{\uetwo}}
      {\lappsubstsrel{\ulsubsts}{\uty}}{j}
    \end{stackTL}\\

    \uvapproxjdg{\ulbGamma}{\uvone}{\uvtwo}{\uty}\defeq
    \begin{stackTL}
      \forall j \geq 0,
      \inlrelG{\lemptystore}{\lemptystore}{\ulsubsts}{\ulbGamma}{j}.
      \inurelV{\lappsubstsone{\ulsubsts}{\uvone}}
      {\lappsubststwo{\ulsubsts}{\uvtwo}}
      {\lappsubstsrel{\ulsubsts}{\uty}}{j}
    \end{stackTL}\\

    \Lapproxjdg{\ulGamma}{\lstoreone}{\leone}{\lstoretwo}{\letwo}{\lty} \defeq
    \begin{stackTL}
      \forall j \geq 0,
      \inlrelG{\lstoreonepr}{\lstoretwopr}{\ulsubsts}{\ulGamma}{j} . \\
      \quad\begin{stackTL}
      \inlrelE{ \lstorejoin \lstoreone \lstoreonepr}{\lappsubstsone{\lsubsts}{\leone}}
              {\lstorejoin \lstoretwo \lstoretwopr}{\lappsubststwo{\lsubsts}{\letwo}}
              {\lappsubstsrel{\lsubsts}{\lty}}{j}
      \end{stackTL}
    \end{stackTL}
  \end{smathpar}
  \caption{Logical Approximation for Open Terms}
\end{figure}

\begin{figure}
  \begin{smathpar}
    \ucapproxjdg{\ulbGamma}{\ueone}{\uetwo}{\uty} \defeq \forall \uctxt.~ \ujudg{\uempty}{\uctxt\hw{\ueone}}{\uunitty} \wedge \ujudg{\uempty}{\uctxt\hw{\uetwo}}{\uunitty} \wedge \uctxt\hw{\ueone} \uredexstepstar \uunite \implies \uctxt\hw{\uetwo} \uredexstepstar \uunite\\

    \Lcapproxjdg{\ulGamma}{\lstoreone}{\leone}{\lstoretwo}{\letwo}{\lty}
    \defeq \forall \uctxt.~
    \ujudg{\uempty}{\uctxt\hw{\hconf{\lstoreone}{\leone}}}{\uunitty}
    \wedge
    \ujudg{\uempty}{\uctxt\hw{\ueone}{\hconf{\lstoretwo}{\letwo}}}{\uunitty}
    \wedge \uctxt\hw{\hconf{\lstoreone}{\leone}} \uredexstepstar
    \uunite \implies \uctxt\hw{\hconf{\lstoretwo}{\letwo}}
    \uredexstepstar \uunite\\
  \end{smathpar}
  \caption{Contextual Approximation}
\end{figure}

The Fundamental Property is the key to proving parametricity results.

\begin{lemma}[Fundamental Property]
  \label{proofs:fundamental}
  {~}
  \begin{enumerate}
  \item If $\ujudg{\ulbGamma}{\uv}{\uty}$ then $\uvapproxjdg{\ulbGamma}{\uv}{\uv}{\uty}$
  \item If $\ujudg{\ulbGamma}{\ue}{\uty}$ then $\uapproxjdg{\ulbGamma}{\ue}{\ue}{\uty}$
  \item If $\ljudg{\lstorety}{\ulGamma}{\lstore}{\le}{\lty}$ then $\Lapproxjdg{\ulbGamma}{\lstore}{\le}{\lstore}{\le}{\lty}$
  \end{enumerate}
\end{lemma}

Finally, we prove our logical relation is sound with respect to
contextual equivalence, that is, it can be used as a more tractible
way to prove contextual equivalence results, such as lump/unlump
cancellation.
\begin{theorem}[Soundness of Logical Relation]
  \label{proofs:soundness}
  {~}
  $\lapproxsymbol \subset \capproxsymbol$
\end{theorem}

The proof is by induction on contexts, showing that every term
formation rule preserves logical relatedness. These ``compatibility''
lemmas are extensive, but their proofs are simple.  Their proofs are
in the extended technical report.

\end{version}
\end{document}